\documentclass[11pt]{article}

\usepackage{amssymb,amsmath,amsthm}
\usepackage{graphicx}
\usepackage{float}
\usepackage{subfig}
\usepackage[numbers]{natbib}
\usepackage{calrsfs}
\usepackage[margin=1in]{geometry}
\usepackage{setspace}
\usepackage{hyperref}

\newtheorem{theorem}{Theorem}[section]

\newtheorem{proposition}[theorem]{Proposition}
\newtheorem{corollary}[theorem]{Corollary}
\theoremstyle{definition}
\newtheorem{definition}[theorem]{Definition}
\theoremstyle{remark}

\newcommand{\Lsym}{\mathcal{L}}

\title{Boundary-Value Friedkin–Johnsen Dynamics and Influence-Based Centralities on Networks}

\author{Moses Boudourides \\[0.5em]
{\small School of Professional Studies, Northwestern University} \\
{\small \texttt{Moses.Boudourides@northwestern.edu}}}

\date{}

\begin{document}

\doublespacing

\maketitle

\begin{abstract}
Understanding which nodes are most able to transmit or receive influence is a central problem in networked opinion dynamics, yet topology-only centrality measures do not account for heterogeneous susceptibility to social influence. We study source-specific, susceptibility-dependent influence in the Friedkin--Johnsen model on directed weighted networks. By treating fully stubborn agents as boundary nodes and the remaining agents as interior nodes, we formulate the dynamics as a discrete boundary-value problem. This formulation yields a Green operator and an all-source steady-state response matrix that quantify how a fixed opinion at each source propagates to every target. Building on this response matrix, we define influence-based broadcasting and reception measures that distinguish structural position from model-implied influence. We establish conditions for existence and uniqueness of the steady state, derive transient and steady-state representations, and characterize sensitivity to susceptibility parameters and perturbations of the influence network. For the undirected homogeneous-susceptibility specialization, we also provide a corresponding spectral analysis. Computational experiments on benchmark networks illustrate how heterogeneous susceptibility profiles shape response-based rankings and how these rankings differ from topology-only centralities. The framework provides a basis for studying influence propagation, source selection, and susceptibility-aware network intervention.
\end{abstract}

\noindent\textbf{Keywords:} Friedkin--Johnsen model; opinion dynamics; social influence; influence propagation; influence measures; heterogeneous susceptibility; weighted directed networks; boundary-value problem; Green operator; applied network science.

\section{Introduction}

\subsection{Opinion dynamics and applied network science}

Opinion dynamics provides a network-science framework for studying how attitudes, beliefs, and preferences evolve through interpersonal interaction. A recurring practical question is which nodes are able to transmit influence most effectively and which nodes are most receptive to influence from externally fixed sources. Classical centrality measures quantify structural prominence from network topology, but they do not generally account for agents' heterogeneous susceptibility to social influence, their attachment to initial opinions, or the dynamical pathway through which a source affects a target.

The Friedkin--Johnsen (FJ) model is a foundational model for this setting. It represents each agent's opinion as a balance between an initial private belief and interpersonal influence from neighbouring agents \citep{FJ1990,Friedkin1998,FJ1999,FJ2011}. Our principal mathematical object is an all-source steady-response matrix $R=(R_{ab})$: its $(a,b)$ entry is the steady response of target $b$ to a unit fixed-opinion intervention at source $a$. From $R$, we derive source-side broadcasting, target-side reception, and, where used, dynamic response descriptors. These are model-implied, intervention-oriented quantities under the specified FJ model.

Treating fixed sources as boundary nodes enables a common comparison of candidate broadcasters, target-level reception, and source-specific reach under one explicitly stated dynamical model.

\subsection{Related work}

The FJ model extends the classical DeGroot averaging model by allowing agents to retain attachment to their initial opinions \citep{DeGroot1974,FJ1990,Friedkin1998,FJ1999,FJ2011}. Existing work establishes equilibrium, stubborn-agent, and convergence properties in networked opinion dynamics \citep{Proskurnikov2017,GhaderiSrikant2014}, while related models examine the tension between private opinions and interpersonal influence \citep{BindelKleinbergOren2015}. A complementary literature studies social power and reflected appraisal in evolving influence networks \citep{JiaMirTabatabaeiFriedkinBullo2015,TianEtAl2022,Friedkin2011ReflectedAppraisal}. Our focus is narrower: for a fixed network and susceptibility profile, we construct the steady response produced at every target by each fixed source.

Centrality interpretation depends on the flow process and analytical objective it represents \citep{Borgatti2005,Borgatti2006,BorgattiEverett2006,BoldiVigna2014}. Walk-based communicability, dynamic communicability, and random-walk diffusion likewise provide process-dependent perspectives on reach and ranking \citep{EstradaHatano2008,GrindrodEtAl2011,MasudaPorterLambiotte2017}. The boundary-value viewpoint used here is a modeling analogy, linked to finite-network Dirichlet and absorbing-chain representations \citep{DoyleSnell1984,KemenySnell1960}, not a claim that FJ trajectories are ordinary random walks.

The contribution is an all-source response construction that combines fixed boundary sources, heterogeneous susceptibility, and multistep network propagation. Proposition~\ref{prop:response-resolvent-factorization} then identifies its exact relationship to topology-only resolvent centralities, while the resulting broadcasting and reception measures are evaluated against structural baselines.

\subsection{Contributions and paper organization}

We formulate discrete-time FJ dynamics with fully stubborn sources as a network boundary-value problem, derive transient and steady-state Green-operator representations, and analyze convergence, susceptibility sensitivity, and perturbations of the influence network. We also give an all-source response construction and characterize its precise relation to the FJ resolvent and to a homogeneous Katz-type special case. The resulting target-side influenceability descriptors, source-side broadcasting scores, and target-side reception scores are evaluated against structural baselines across eight observed, benchmark, and controlled synthetic networks. The evaluation reports both agreement and divergence, including top-ten overlap, a model-based source-selection diagnostic, and the effect of matched heterogeneous versus homogeneous susceptibility profiles.

The remainder of the article is organized as follows. Section~2 introduces the graph notation, the FJ boundary-value formulation, and the associated theoretical results. Section~3 specifies the source-indexed computational formulation and all-vertex response calculation. Section~4 defines target-side influenceability descriptors and the dynamics-induced broadcasting and reception measures. Section~5 presents the computational evaluation on benchmark and case-study networks. Section~6 discusses interpretation, limitations, and prospective applications. Section~7 concludes. Supplementary appendices provide technical proofs, computational details, and additional robustness analyses.

\section{Preliminaries: Graph Notation and Influence Dynamics}
\label{sec:preliminaries}

\subsection{Direction convention, state space, and standing assumptions}
\label{subsec:direction-convention}

Let $V=\{1,\ldots,n\}$ be a finite set of agents, and let
$W\in\mathbb{R}^{n\times n}$ be a nonnegative row-substochastic matrix:
\[
W_{ij}\geq 0,
\qquad
\sum_{j\in V}W_{ij}\leq 1
\quad\text{for every }i\in V.
\]
We use a \emph{receive-from convention}: the entry $W_{ij}>0$ means that the opinion of agent $j$ enters the update of agent $i$. Accordingly, when arrows represent the propagation of influence, the directed support graph is
\[
G_W=(V,E_W),
\qquad
(j,i)\in E_W
\quad\Longleftrightarrow\quad
W_{ij}>0.
\]
Thus, an arc $j\to i$ represents potential propagation of influence from $j$ to $i$. In the empirical preprocessing, every row with at least one received edge is normalized to sum to one. A vertex with no receive-from neighbours is retained with an all-zero row; no self-loop is inserted to repair that row. Thus the row-stochastic setting is the special case in which every row has positive total received weight. Self-loops are allowed whenever $W_{ii}>0$.

Let
\[
S=\operatorname{diag}(s_1,\ldots,s_n),
\qquad 0\leq s_i\leq 1,
\]
be the susceptibility matrix. We partition the agents into the boundary set
\[
\partial\Omega=\{i\in V:s_i=0\}
\]
and the interior set
\[
\Omega=V\setminus\partial\Omega=\{i\in V:s_i>0\}.
\]
Boundary agents are fully stubborn: their opinions remain fixed at their initial values. Interior agents are malleable, but an agent with $0<s_i<1$ remains partially anchored to its initial opinion; it is therefore not, in general, fully susceptible to social influence. We assume that $\partial\Omega$ is nonempty whenever a boundary-value formulation is considered.

For $i\in\Omega$, boundary accessibility means that there is a boundary node $b\in\partial\Omega$ and a directed path $b\to\cdots\to i$ in $G_W$. Equivalently, in the row-transition orientation associated with the matrix $W$, node $i$ can reach a row that receives positive weight from the boundary. Additional reachability, positivity, and norm assumptions required by particular results are stated explicitly below.

\subsection{Undirected reversible specialization and Laplacians}
\label{subsec:undirected-specialization}

The general framework permits directed and asymmetric influence matrices. An undirected support graph only means that $W_{ij}>0$ if and only if $W_{ji}>0$; it does \emph{not} by itself imply that $W$ is symmetric. The spectral results below require the stronger reversible random-walk specialization. Specifically, let $A=A^\top\geq 0$ be the weighted adjacency matrix of an undirected graph with positive degrees
\[
d_i=\sum_{j\in V}A_{ij},
\qquad
D=\operatorname{diag}(d_1,\ldots,d_n),
\]
and set
\[
N=D^{-1}A.
\]
Then $N$ is row-stochastic and is reversible with respect to the degree measure. The associated combinatorial and symmetric normalized Laplacians are
\[
L=D-A,
\qquad
\Lsym=D^{-1/2}LD^{-1/2}
=I-D^{-1/2}AD^{-1/2}.
\]
For $\Omega\subseteq V$, the Dirichlet restriction of the symmetric normalized Laplacian is the principal submatrix $\Lsym_{\Omega\Omega}$ \citep{Chung1997}.

The undirected reversible specialization is introduced only to obtain interpretable spectral statements. It is not used to justify results for general directed or heterogeneous influence networks.

\subsection{The DeGroot and Friedkin--Johnsen models}
\label{subsec:degroot-fj}

In the row-stochastic special case, the update
\begin{equation}
v^{t+1}=Wv^t,
\label{eq:degroot-model}
\end{equation}
is the classical DeGroot averaging model \citep{DeGroot1974}; the row-substochastic convention used below accommodates retained zero-receive rows. Here $v^t\in\mathbb{R}^n$ is the vector of opinions at time $t$. The Friedkin--Johnsen (FJ) model extends this update by allowing agents to retain attachment to their initial opinions \citep{FJ1990,Friedkin1998,FJ1999,FJ2011}:
\begin{equation}
v^{t+1}=SWv^t+(I-S)v^0.
\label{eq:FJ_model}
\end{equation}
The parameter $s_i$ quantifies the openness of agent $i$ to interpersonal influence: $s_i=0$ yields a fully stubborn agent, $0<s_i<1$ yields a partially susceptible agent, and $s_i=1$ yields an agent whose update is entirely determined by network influence.

The FJ model has been studied extensively in connection with stubborn agents, equilibrium behavior, convergence, and evolving social power \citep{GhaderiSrikant2014,JiaMirTabatabaeiFriedkinBullo2015,TianEtAl2022}. Our boundary-value formulation fixes the opinions of the fully stubborn agents and analyzes their propagation through the interior under a fixed influence network and a specified susceptibility profile.

\subsection{Boundary-value formulation and block-matrix reduction}
\label{subsec:ibvp}

Write the initial interior opinions as $v^0|_\Omega=\phi$ and the fixed boundary opinions as $v^0|_{\partial\Omega}=\psi$. Because $s_i=0$ for $i\in\partial\Omega$, equation~\eqref{eq:FJ_model} implies
\[
v^t|_{\partial\Omega}=\psi
\qquad\text{for all }t\geq 0.
\]
After ordering the vertices so that $V=\Omega\cup\partial\Omega$, we write
\begin{equation}
W=
\begin{pmatrix}
W_{\Omega\Omega} & W_{\Omega\partial}\\
W_{\partial\Omega} & W_{\partial\partial}
\end{pmatrix},
\qquad
S=
\begin{pmatrix}
S_\Omega & 0\\
0 & 0
\end{pmatrix},
\qquad
v^t=
\begin{pmatrix}
v_\Omega^t\\
\psi
\end{pmatrix}.
\label{eq:block-partition}
\end{equation}
Here, $W_{\Omega\Omega}\in\mathbb{R}^{|\Omega|\times|\Omega|}$ describes endogenous influence among interior agents, while $W_{\Omega\partial}\in\mathbb{R}^{|\Omega|\times|\partial\Omega|}$ records channels through which boundary agents inject opinions into the interior. The vectors satisfy $\phi,v_\Omega^t\in\mathbb{R}^{|\Omega|}$ and $\psi\in\mathbb{R}^{|\partial\Omega|}$.

\begin{theorem}[Interior affine dynamics]
\label{thm:interior-affine}
The interior opinions satisfy the affine recursion
\begin{equation}
v_\Omega^{t+1}
=S_\Omega W_{\Omega\Omega}v_\Omega^t
+S_\Omega W_{\Omega\partial}\psi
+(I_\Omega-S_\Omega)\phi.
\label{eq:interior_dynamics}
\end{equation}
\end{theorem}

\begin{proof}
Substitution of \eqref{eq:block-partition} into \eqref{eq:FJ_model} gives
\[
\begin{pmatrix}
v_\Omega^{t+1}\\
\psi
\end{pmatrix}
=
\begin{pmatrix}
S_\Omega\bigl(W_{\Omega\Omega}v_\Omega^t+W_{\Omega\partial}\psi\bigr)\\
0
\end{pmatrix}
+
\begin{pmatrix}
(I_\Omega-S_\Omega)\phi\\
\psi
\end{pmatrix},
\]
whose upper block is \eqref{eq:interior_dynamics}.
\end{proof}

From a network perspective, the block $W_{\Omega\Omega}$ governs the circulation of opinions among malleable agents, whereas $W_{\Omega\partial}$ captures the direct pathways through which fixed boundary opinions propagate into the interior. The term $(I_\Omega-S_\Omega)\phi$ represents the continuing pull of interior agents' private initial opinions. Thus, a boundary node may affect an interior target both directly and through multistep propagation across the network.

\subsection{Resolvent formulation, contraction, and perturbations}
\label{subsec:resolvent}

Define the interior iteration matrix
\[
M:=S_\Omega W_{\Omega\Omega}
\]
and the affine input
\[
b:=S_\Omega W_{\Omega\partial}\psi+(I_\Omega-S_\Omega)\phi.
\]
Then equation~\eqref{eq:interior_dynamics} becomes $v_\Omega^{t+1}=Mv_\Omega^t+b$.

\begin{proposition}[Substochastic contraction criterion]
\label{prop:substochastic-contraction}
The matrix $M$ is nonnegative and substochastic. If, in the row-transition graph of $M$, every interior node can reach a row $i$ satisfying
\[
\sum_{j\in\Omega}M_{ij}<1,
\]
then $\rho(M)<1$. In particular, this condition holds if every interior node is boundary-accessible in the propagation graph $G_W$.
\end{proposition}

\begin{proof}
This is the standard finite substochastic contraction argument \citep{Seneta2006}. 
Since $W$ is row-substochastic and $0\leq s_i\leq 1$, each row sum of $M$ is at most one. For a finite substochastic matrix, the stated reachability condition excludes a closed communicating class whose row sums are all one. Equivalently, there is an integer $m\geq 1$ for which $M^m\mathbf{1}<\mathbf{1}$ componentwise. Hence $\|M^m\|_\infty<1$, and therefore $\rho(M)^m=\rho(M^m)\leq\|M^m\|_\infty<1$ \citep{Seneta2006}.

If an interior node is boundary-accessible in $G_W$, then the reversed path in the row-transition orientation reaches an interior row that receives positive weight from the boundary. Such a row has $\sum_{j\in\Omega}M_{ij}<1$. The first condition therefore follows.
\end{proof}

\begin{theorem}[Existence and uniqueness of the steady state]
\label{thm:steady-state}
Assume that
\begin{equation}
\rho(M)=\rho(S_\Omega W_{\Omega\Omega})<1.
\label{eq:spectral_radius_condition}
\end{equation}
Then there exists a unique steady-state opinion vector
\begin{equation}
v_\Omega^*
=(I_\Omega-S_\Omega W_{\Omega\Omega})^{-1}
\Bigl(S_\Omega W_{\Omega\partial}\psi+(I_\Omega-S_\Omega)\phi\Bigr).
\label{eq:steady_state}
\end{equation}
\end{theorem}

\begin{proof}
Condition \eqref{eq:spectral_radius_condition} implies that $I_\Omega-M$ is invertible and that the Neumann series $\sum_{k=0}^{\infty}M^k$ converges. Solving the fixed-point equation $v_\Omega^*=Mv_\Omega^*+b$ yields \eqref{eq:steady_state}.
\end{proof}

\begin{definition}[Green operator]
\label{def:green-operator}
Under \eqref{eq:spectral_radius_condition}, define
\begin{equation}
G_S:=(I_\Omega-S_\Omega W_{\Omega\Omega})^{-1}
=\sum_{k=0}^{\infty}(S_\Omega W_{\Omega\Omega})^k.
\label{eq:green-operator}
\end{equation}
\end{definition}

From a network perspective, the Green operator encodes the cumulative effect of all weighted propagation pathways connecting the fixed boundary to interior nodes. Its $k$th Neumann term $(S_\Omega W_{\Omega\Omega})^k$ records propagation through exactly $k$ interior-to-interior update steps, with every visit attenuated by the relevant susceptibility and edge weight. Thus the steady state combines direct boundary injection, repeated interior propagation, and persistence of interior initial opinions. This is analogous to a fundamental matrix with the boundary treated as absorbing \citep{KemenySnell1960,DoyleSnell1984}, but retains the FJ model's initial-opinion anchoring and heterogeneous susceptibility.

\begin{theorem}[Transient solution]
\label{thm:transient-solution}
For every $t\in\mathbb{N}_0$,
\begin{equation}
v_\Omega^t
=M^t\phi+\sum_{k=0}^{t-1}M^k b,
\label{eq:transient-solution}
\end{equation}
where the empty sum is zero when $t=0$.
\end{theorem}

\begin{proof}
The result follows by iterating the affine recursion $v_\Omega^{t+1}=Mv_\Omega^t+b$ from $v_\Omega^0=\phi$.
\end{proof}

\begin{proposition}[Error recursion and convergence bound]
\label{prop:error-recursion}
Assume \eqref{eq:spectral_radius_condition}. Then
\begin{equation}
v_\Omega^t-v_\Omega^*=M^t(\phi-v_\Omega^*)
\label{eq:error-exact}
\end{equation}
for every $t\in\mathbb{N}_0$. Consequently, for any submultiplicative matrix norm and its compatible vector norm,
\begin{equation}
\|v_\Omega^t-v_\Omega^*\|
\leq
\|M^t\|\,\|\phi-v_\Omega^*\|.
\label{eq:error-norm}
\end{equation}
\end{proposition}

\begin{proof}
Subtract the fixed-point identity $v_\Omega^*=Mv_\Omega^*+b$ from $v_\Omega^{t+1}=Mv_\Omega^t+b$, and iterate the resulting linear error recursion.
\end{proof}

\begin{theorem}[Quantitative convergence rate]
\label{thm:rate-spectral-radius}
Assume \eqref{eq:spectral_radius_condition}. For every $\delta$ satisfying $0<\delta<1-\rho(M)$, there exists $C_\delta\geq 1$ such that
\begin{equation}
\|v_\Omega^t-v_\Omega^*\|
\leq
C_\delta\bigl(\rho(M)+\delta\bigr)^t\,\|\phi-v_\Omega^*\|
\label{eq:rate-Cdelta}
\end{equation}
for all $t\in\mathbb{N}_0$. Moreover,
\begin{equation}
\limsup_{t\to\infty}\|v_\Omega^t-v_\Omega^*\|^{1/t}
\leq\rho(M).
\label{eq:gelfand}
\end{equation}
\end{theorem}

\begin{proof}
By Proposition~\ref{prop:error-recursion}, it suffices to bound $\|M^t\|$. Gelfand's formula gives $\lim_{t\to\infty}\|M^t\|^{1/t}=\rho(M)$ for every submultiplicative norm \citep{HornJohnson2013,Higham2008}. The usual finite-prefix argument then yields a constant $C_\delta$ for which $\|M^t\|\leq C_\delta(\rho(M)+\delta)^t$ for all $t$.
\end{proof}

In terms of opinion dynamics, the convergence rate is governed by the persistence of opinion circulation within the interior network. Strong attachment to initial opinions, boundary injection, or both reduce the effective persistence of this circulation and can accelerate stabilization. The result concerns convergence under the specified FJ model; it should not be interpreted as an empirical estimate of the speed of real-world social influence without appropriate data and model validation.

\begin{theorem}[Sensitivity to an interior susceptibility parameter]
\label{thm:sensitivity-s}
Let $k\in\Omega$, let $E_k=e_ke_k^\top$, and suppose that $s_k$ varies over an open interval on which \eqref{eq:spectral_radius_condition} continues to hold. Then $v_\Omega^*$ is differentiable with respect to $s_k$, and
\begin{equation}
\frac{\partial v_\Omega^*}{\partial s_k}
=G_S E_k\Bigl(W_{\Omega\Omega}v_\Omega^*+W_{\Omega\partial}\psi-\phi\Bigr).
\label{eq:sensitivity}
\end{equation}
Equivalently,
\begin{equation}
\frac{\partial v_\Omega^*}{\partial s_k}
=\Bigl(W_{\Omega\Omega}v_\Omega^*+W_{\Omega\partial}\psi-\phi\Bigr)_kG_Se_k.
\label{eq:sensitivity-column}
\end{equation}
\end{theorem}

\begin{proof}
Write $A(s)=I_\Omega-S_\Omega W_{\Omega\Omega}$ and $b(s)=S_\Omega W_{\Omega\partial}\psi+(I_\Omega-S_\Omega)\phi$, so that $v_\Omega^*=A(s)^{-1}b(s)$. Since $\partial S_\Omega/\partial s_k=E_k$, differentiation of $A(s)^{-1}b(s)$ gives \eqref{eq:sensitivity} and \eqref{eq:sensitivity-column}.
\end{proof}

\begin{theorem}[Perturbation bound for the influence matrix]
\label{thm:perturbation-W}
Fix $S_\Omega$, $\phi$, and $\psi$. Let $W$ and $\widetilde W$ be nonnegative row-substochastic influence matrices with the same partition $V=\Omega\cup\partial\Omega$, and suppose that both associated interior matrices satisfy \eqref{eq:spectral_radius_condition}. Define
\[
A=I_\Omega-S_\Omega W_{\Omega\Omega},
\qquad
\widetilde A=I_\Omega-S_\Omega\widetilde W_{\Omega\Omega},
\]
\[
b=S_\Omega W_{\Omega\partial}\psi+(I_\Omega-S_\Omega)\phi,
\qquad
\widetilde b=S_\Omega\widetilde W_{\Omega\partial}\psi+(I_\Omega-S_\Omega)\phi.
\]
Then, for every compatible submultiplicative norm,
\begin{align}
\|v_\Omega^*(W)-v_\Omega^*(\widetilde W)\|
\leq{}&
\|A^{-1}\|\,\|S_\Omega\|\,
\|W_{\Omega\Omega}-\widetilde W_{\Omega\Omega}\|\,
\|\widetilde A^{-1}\|\,\|\widetilde b\| \notag\\
&+
\|\widetilde A^{-1}\|\,\|S_\Omega\|\,
\|W_{\Omega\partial}-\widetilde W_{\Omega\partial}\|\,\|\psi\|.
\label{eq:perturbation}
\end{align}
\end{theorem}

\begin{proof}
Use the resolvent identity
\[
A^{-1}-\widetilde A^{-1}
=A^{-1}(\widetilde A-A)\widetilde A^{-1}
\]
and the decomposition
\[
A^{-1}b-\widetilde A^{-1}\widetilde b
=(A^{-1}-\widetilde A^{-1})\widetilde b
+A^{-1}(b-\widetilde b).
\]
Submultiplicativity and $b-\widetilde b=S_\Omega(W_{\Omega\partial}-\widetilde W_{\Omega\partial})\psi$ give \eqref{eq:perturbation} \citep{HornJohnson2013,Higham2008}.
\end{proof}

\subsection{Spectral specialization for homogeneous susceptibility}
\label{subsec:spectral-specialization}

Assume the undirected reversible random-walk specialization $W=N=D^{-1}A$ from Section~\ref{subsec:undirected-specialization}, and let $S_\Omega=sI_\Omega$ with $s\in(0,1)$. Define
\[
B_\Omega
:=D_\Omega^{1/2}N_{\Omega\Omega}D_\Omega^{-1/2}
=D_\Omega^{-1/2}A_{\Omega\Omega}D_\Omega^{-1/2}
=I_\Omega-\Lsym_{\Omega\Omega}.
\]
The matrix $B_\Omega$ is symmetric. Let $(\lambda_k,u_k)$, $k=1,\ldots,|\Omega|$, denote its real eigenpairs, with $\{u_k\}$ orthonormal in the Euclidean inner product.

\begin{corollary}[Corrected spectral representation]
\label{cor:spectral}
The Green operator has the representation
\begin{equation}
G_s:=(I_\Omega-sN_{\Omega\Omega})^{-1}
=D_\Omega^{-1/2}
\left(
\sum_{k=1}^{|\Omega|}
\frac{1}{1-s\lambda_k}\,u_ku_k^\top
\right)
D_\Omega^{1/2}.
\label{eq:spectral-green}
\end{equation}
\end{corollary}

\begin{proof}
The similarity relation
\[
I_\Omega-sN_{\Omega\Omega}
=D_\Omega^{-1/2}(I_\Omega-sB_\Omega)D_\Omega^{1/2}
\]
implies the corresponding inverse relation. Since $B_\Omega$ is symmetric, $(I_\Omega-sB_\Omega)^{-1}$ admits the stated orthonormal eigen-expansion, which yields \eqref{eq:spectral-green}.
\end{proof}

\begin{corollary}[Degree-weighted convergence bound]
\label{cor:rate-dirichlet-spectrum}
Let
\[
\|x\|_{D_\Omega}:=\bigl(x^\top D_\Omega x\bigr)^{1/2}
\]
and let $\lambda_\star=\max_{1\leq k\leq|\Omega|}|\lambda_k|=\rho(N_{\Omega\Omega})$, where $\lambda_k$ denotes the $k$th eigenvalue of the symmetric normalized adjacency matrix $B_\Omega$ introduced above. Then
\begin{equation}
\|v_\Omega^t-v_\Omega^*\|_{D_\Omega}
\leq
(s\lambda_\star)^t\,
\|\phi-v_\Omega^*\|_{D_\Omega}
\label{eq:rate-dirichlet}
\end{equation}
for every $t\in\mathbb{N}_0$, and $s\lambda_\star<1$.
\end{corollary}

\begin{proof}
The similarity relation gives
\[
D_\Omega^{1/2}N_{\Omega\Omega}^t
=B_\Omega^tD_\Omega^{1/2}.
\]
Because $B_\Omega$ is symmetric, $\|B_\Omega^t\|_2=\lambda_\star^t$. Applying this identity to the exact error recursion \eqref{eq:error-exact} gives \eqref{eq:rate-dirichlet}.
\end{proof}

Interpreted as a network process, the spectral specialization supplies an interpretable description of how the undirected network modes are attenuated by a common susceptibility parameter. It is a special case used for intuition and computation; the resolvent formulation in Section~\ref{subsec:resolvent} remains the primary framework for directed and heterogeneous networks.

From a computational standpoint, the same resolvent structure that underpins the theoretical analysis also provides the foundation for efficient computation. In particular, the source-indexed response quantities can be evaluated by solving sparse linear systems involving $I_\Omega-S_\Omega W_{\Omega\Omega}$ rather than by forming a dense inverse explicitly. Section~\ref{sec:computational-formulation} specifies the all-source scanning algorithm, numerical complexity, and reproducibility requirements.

\section{Computational Formulation and Algorithmic Implementation}
\label{sec:computational-formulation}

\subsection{Computing source-indexed steady-state responses}
\label{subsec:source-response-scope}

This section fixes the source-indexed counterfactual experiments used to compute the response matrix. Target-side influenceability descriptors and source-side broadcasting scores are interpreted separately in Section~\ref{sec:influenceability}; here the focus is the reproducible calculation that supplies their inputs.

For every candidate source $a\in V$, we define a source-indexed FJ experiment. Fix a source value $q=1$ and a baseline susceptibility vector $s\in[0,1]^V$. Let
\begin{equation}
\label{eq:source-indexed-susceptibility}
s^{[a]}_k=
\begin{cases}
0, & k=a,\\
s_k, & k\neq a,
\end{cases}
\qquad k\in V,
\end{equation}
and let the source-indexed initial condition be
\begin{equation}
\label{eq:source-indexed-initial-condition}
v_k^{[a]}(0)=
\begin{cases}
1, & k=a,\\
0, & k\neq a.
\end{cases}
\end{equation}
Thus, source $a$ is converted into a fully stubborn unit-opinion source. Any other baseline-stubborn node remains stubborn but is assigned the neutral value zero. For each source-induced boundary/interior partition, the associated steady response is computed by solving the corresponding stable linear system; no dense inverse is required.

The source-indexed construction is a repeated \emph{model-based intervention experiment}: the dynamical system is modified for each choice of $a$. Consequently, every row of the all-source response matrix calculated below is a mathematical counterfactual under the specified FJ model, fixed network, and chosen susceptibility profile.

Throughout this section, the influence-arrow convention from Section~\ref{subsec:direction-convention} is retained. An arrow $a\to b$ means that the opinion of $a$ enters the update of $b$, equivalently $W_{ba}>0$. Rows of the all-source response matrix therefore index sources, columns index targets, and an entry in row $a$, column $b$ measures the effect of source $a$ on target $b$.

\subsection{All-vertex scanning and response matrices}
\label{subsec:all-vertex-scanning}

The fixed-boundary construction can be extended to every vertex by repeating the source-indexed experiment in Section~\ref{subsec:source-response-scope}. For a source $a$, let
\[
B_a:=\{k\in V:s_k^{[a]}=0\},
\qquad
\Omega_a:=V\setminus B_a,
\]
and let $\psi^{[a]}\in\mathbb{R}^{|B_a|}$ equal one at source $a$ and zero at every other boundary vertex. Define
\begin{equation}
\label{eq:all-source-response}
\mathcal U^\infty_{ab}:=
\begin{cases}
1, & b=a,\\
\left[(I-S_{\Omega_a}^{[a]}W_{\Omega_a\Omega_a})^{-1}
S_{\Omega_a}^{[a]}W_{\Omega_a B_a}\psi^{[a]}\right]_b, & b\in\Omega_a,\\
0, & b\in B_a\setminus\{a\},
\end{cases}
\end{equation}
whenever the source-indexed system is stable. Let
\[
R:=\mathcal U^\infty\in[0,1]^{n\times n}
\]
denote the resulting all-source steady-response matrix. Its $a$th row is the target-response profile generated by the intervention at source $a$.

\begin{proposition}[Response--resolvent factorization and homogeneous specialization]
\label{prop:response-resolvent-factorization}
Assume that every baseline susceptibility satisfies $s_i\in(0,1)$, so that the only boundary node in the source-indexed experiment for $a$ is $a$ itself, and let
\[
S:=\operatorname{diag}(s),
\qquad
G:=(I-SW)^{-1}.
\]
Then the all-source steady-response matrix satisfies
\begin{equation}
\label{eq:response-resolvent-factorization}
R_{ab}=\frac{G_{ba}}{G_{aa}},
\qquad
R=D_G^{-1}G^\top,
\qquad
D_G:=\operatorname{diag}(G_{11},\ldots,G_{nn}).
\end{equation}
Consequently, $R$ is a source-normalized, transposed FJ resolvent. In particular, unless $G_{11}=\cdots=G_{nn}$, it cannot be a global scalar multiple of $G^\top$.

If $S=\beta I$ and the diagonal of $G=(I-\beta W)^{-1}$ is constant, say $G_{aa}=c$ for every $a$, then
\begin{equation}
\label{eq:response-homogeneous-resolvent}
R=c^{-1}(I-\beta W^\top)^{-1}.
\end{equation}
Thus, if $\kappa:=(I-\beta W^\top)^{-1}\mathbf{1}-\mathbf{1}$ denotes the source-oriented Katz vector based on the receive-from matrix, the broadcasting vector $B$ obeys
\begin{equation}
\label{eq:broadcasting-katz-specialization}
B=\frac{\kappa+(1-c)\mathbf{1}}{c(n-1)}.
\end{equation}
\end{proposition}

\begin{proof}
Fix a source $a$. The vector $u^{[a]}$ defined by $u^{[a]}_a=1$ and $u^{[a]}_b=R_{ab}$ for $b\neq a$ solves the source-indexed steady-state equations. Since $G=(I-SW)^{-1}$, the vector $G e_a/G_{aa}$ equals one at $a$ and satisfies the same interior equations at every $b\neq a$, because
\[
(I-SW)\frac{G e_a}{G_{aa}}=\frac{e_a}{G_{aa}}.
\]
Uniqueness of the stable source-indexed system gives $u^{[a]}=G e_a/G_{aa}$, proving \eqref{eq:response-resolvent-factorization}. If $R=\alpha G^\top$ for a scalar $\alpha$, then $1=R_{aa}=\alpha G_{aa}$ for every $a$, so equality of all diagonal entries of $G$ is necessary. Under the stated homogeneous condition, $D_G=cI$, which gives \eqref{eq:response-homogeneous-resolvent}. Summing each row of $R$, excluding its unit diagonal entry, then yields \eqref{eq:broadcasting-katz-specialization}.
\end{proof}

\begin{corollary}[Criterion for scalar-resolvent reduction]
\label{cor:scalar-resolvent-reduction}
Under the assumptions of Proposition~\ref{prop:response-resolvent-factorization}, the following statements are equivalent:
\begin{enumerate}
\item There exists a scalar $\alpha>0$ such that $R=\alpha G^\top$.
\item The diagonal of $G$ is constant: $G_{11}=\cdots=G_{nn}$.
\end{enumerate}
When these conditions hold, $\alpha=G_{aa}^{-1}$ for every $a$. Hence, in the general case, the all-source response matrix is a source-normalized transposed FJ resolvent rather than a scalar multiple of a resolvent kernel. A Katz-type reduction for broadcasting requires both the homogeneous specialization $S=\beta I$ and the constant-diagonal condition, in which case equation~\eqref{eq:broadcasting-katz-specialization} applies.
\end{corollary}

\begin{proof}
If $R=\alpha G^\top$, then $1=R_{aa}=\alpha G_{aa}$ for every $a$, so all diagonal entries of $G$ are equal. Conversely, if $G_{aa}=c$ for every $a$, then equation~\eqref{eq:response-resolvent-factorization} gives $R=c^{-1}G^\top$.
\end{proof}

The corollary supplies the formal distinction from a scalar-resolvent centrality. Katz-type scores use one common geometric path kernel and a topology-derived propagation matrix, whereas $R$ normalizes each source intervention by its own diagonal Green entry $G_{aa}$ and, in the full construction of equation~\eqref{eq:all-source-response}, can incorporate heterogeneous susceptibilities and multiple baseline-stubborn boundary agents. Moreover, Theorem~\ref{thm:sensitivity-s} shows that response values can change when the susceptibility field changes while $W$ is held fixed; no centrality defined solely from $W$ can reproduce that susceptibility-dependent response family. Communicability instead uses the distinct exponential kernel $\exp(A)$ \citep{EstradaHatano2008}; it is related through path aggregation but is not a finite-parameter specialization of the source-normalized FJ response matrix. The homogeneous constant-diagonal case above states the exceptional setting in which source broadcasting is an affine transformation of a source-oriented Katz score.

\subsection{Computational implementation} 
\label{subsec:response-computation}

The all-vertex scan is a repeated linear-system calculation, not a dense matrix-inversion procedure. For each source $a$, the implementation should construct the source-indexed boundary and interior sets, solve the sparse linear system in \eqref{eq:all-source-response}, and store the resulting row of $R$.

\begin{center}
\fbox{\begin{minipage}{0.90\linewidth}
\textbf{All-source response algorithm.}
\begin{enumerate}
\item Input $W$, the baseline susceptibility vector $s$, source value $1$, numerical tolerance, and any response threshold.
\item For each $a\in V$, set $s_a^{[a]}=0$, set the source boundary value to one, and set all other initial and boundary values to zero.
\item Form the sparse interior operator $A_a=I-S_{\Omega_a}^{[a]}W_{\Omega_a\Omega_a}$ and verify the applicable stability condition.
\item Solve $A_ax_a=S_{\Omega_a}^{[a]}W_{\Omega_aB_a}\psi^{[a]}$ without explicitly forming $A_a^{-1}$.
\item Insert $x_a$ and the prescribed boundary values into row $a$ of $R$; then compute any required timing matrices or source/target summaries.
\item Output $R$, all reported score vectors, solver residuals, parameter values, and random seeds used in repeated simulations.
\end{enumerate}
\end{minipage}}
\end{center}

For sparse networks, direct sparse factorizations or iterative solvers are generally preferable to dense inversion. A complete scan may require one source-indexed solve per candidate source and, in a Monte Carlo design, per source--draw pair. Because converting the source into a boundary node changes the interior system, factorization reuse must be described rather than assumed. The computational report should state matrix sizes and sparsity, solver and preconditioner, stopping tolerance, maximum residual, treatment of nearly singular systems, response threshold, and the precise source-scan procedure. 

\section{Influenceability Measures and Dynamics-Induced Centralities}
\label{sec:influenceability}

The source-indexed calculation of $R$ is specified in Section~\ref{sec:computational-formulation}. This section interprets the resulting target-side descriptors and source- or target-oriented summary scores as distinct dynamics-induced network quantities.

\subsection{Target-side response descriptors}
\label{subsec:target-response-descriptors}

Let $G_W$ be the influence graph from Section~\ref{subsec:direction-convention}. For $a,b\in V$, define the directed influence distance
\begin{equation}
\label{eq:influence-distance}
d_{\mathrm{inf}}(a,b)
:=\min\bigl\{\ell\in\mathbb{N}_0:\text{there is a directed path }a\to\cdots\to b\text{ of length }\ell\text{ in }G_W\bigr\},
\end{equation}
with $d_{\mathrm{inf}}(a,b)=+\infty$ when no such path exists. In particular, $d_{\mathrm{inf}}(a,a)=0$.

\begin{definition}[Arrival, first response, stabilization, and steady response]
\label{def:target-response-descriptors}
For a source $a$ and target $b$, define the \emph{arrival time}
\begin{equation}
\label{eq:arrival-time}
T_b^0(a):=\inf\bigl\{t\in\mathbb{N}_0:v_b^{[a]}(t)>0\bigr\},
\end{equation}
with $\inf\emptyset=+\infty$. When $T_b^0(a)<+\infty$, define the \emph{first response}
\begin{equation}
\label{eq:first-response}
E_{ab}:=v_b^{[a]}\!\bigl(T_b^0(a)\bigr);
\end{equation}
otherwise set $E_{ab}=0$. For $\varepsilon>0$, define the \emph{stabilization time}
\begin{equation}
\label{eq:stabilization-time}
T_b^{\mathrm{stab}}(a;\varepsilon)
:=\inf\left\{t\in\mathbb{N}_0:
\sup_{\tau\geq t}\left|v_b^{[a]}(\tau)-\bar v_b^{[a]}\right|\leq\varepsilon\right\}.
\end{equation}
Finally, define the \emph{steady response}
\begin{equation}
\label{eq:steady-response}
\mathcal U^\infty_{ab}:=\bar v_b^{[a]}.
\end{equation}
\end{definition}

The stabilization time in \eqref{eq:stabilization-time} is a genuine eventual-stability criterion. It differs from a first-entrance time, which may be small even if the trajectory later leaves the prescribed neighborhood of its limit.

\begin{proposition}[Arrival time and influence distance]
\label{prop:arrival-distance}
Fix $a,b\in V$. Suppose that the source-indexed experiment has a single nonzero source as in \eqref{eq:source-indexed-initial-condition}, all propagation weights along at least one shortest directed path from $a$ to $b$ are strictly positive, and every non-source vertex on that path has strictly positive susceptibility. Then
\begin{equation}
\label{eq:arrival-distance-equality}
T_b^0(a)=d_{\mathrm{inf}}(a,b).
\end{equation}
If $b$ is not reachable from $a$, then $T_b^0(a)=+\infty$ and $\mathcal U^\infty_{ab}=0$.
\end{proposition}

\begin{proof}
Under the receive-from convention, influence can traverse at most one influence arrow per time step. Thus no positive response can occur before the length of a shortest directed path. Conversely, nonnegative propagation, positive weights, and positive intermediate susceptibilities ensure a strictly positive contribution along the specified shortest path at its terminal time. If there is no directed path, no term in the iterated update can carry the source value to $b$.
\end{proof}

Interpreted together, the four descriptors serve different purposes. $T_b^0(a)$ records first possible arrival, $E_{ab}$ records the initial positive response, $T_b^{\mathrm{stab}}(a;\varepsilon)$ records eventual stabilization, and $\mathcal U^\infty_{ab}$ records accumulated steady-state response. We use \emph{influenceability} for these target-side response properties. Source-side broadcasting scores are defined separately in Section~\ref{subsec:response-centralities}.

\subsection{Boundary-to-interior response operator}
\label{subsec:boundary-response-operator}

We first consider a fixed boundary set $B=\partial\Omega$ and interior set $\Omega$, with the partition and spectral-radius condition from Section~\ref{subsec:resolvent}. The boundary-to-interior response operator isolates the contribution of boundary values from that of interior initial opinions.

\begin{definition}[Boundary response matrix]
\label{def:boundary-response-matrix}
Let $G_S=(I_\Omega-S_\Omega W_{\Omega\Omega})^{-1}$. The \emph{boundary response matrix} is
\begin{equation}
\label{eq:boundary-response-matrix}
U_B:=G_S S_\Omega W_{\Omega B}
\in\mathbb{R}^{|\Omega|\times|B|}.
\end{equation}
For $j\in\Omega$ and $k\in B$, the entry $(U_B)_{jk}$ is the steady-state response at target $j$ to a unit fixed value at boundary source $k$, with all other boundary values and all interior initial opinions set to zero.
\end{definition}

\begin{proposition}[Boundary-wise steady-state decomposition]
\label{prop:boundary-wise-decomposition}
Under the spectral-radius condition,
\begin{equation}
\label{eq:boundary-wise-decomposition}
v_\Omega^*=U_B\psi+G_S(I_\Omega-S_\Omega)\phi.
\end{equation}
\end{proposition}

\begin{proof}
Substitute \eqref{eq:boundary-response-matrix} into the steady-state formula \eqref{eq:steady_state} and collect the boundary and interior-initial-condition terms.
\end{proof}

\begin{proposition}[Positivity, bounds, and susceptibility monotonicity]
\label{prop:boundary-response-properties}
Suppose $W$ is nonnegative and row-substochastic, $0\leq S_\Omega\leq I_\Omega$, and $\rho(S_\Omega W_{\Omega\Omega})<1$. Then $U_B\geq0$ entrywise and
\begin{equation}
\label{eq:boundary-response-bounds}
0\leq (U_B)_{jk}\leq1.
\end{equation}
Moreover, for every $j\in\Omega$,
\begin{equation}
\label{eq:boundary-response-row-sum}
0\leq\sum_{k\in B}(U_B)_{jk}\leq1.
\end{equation}
The left inequality is strict if $j$ is reachable from at least one boundary source in the influence-arrow convention through a directed path whose nonboundary vertices have positive susceptibility.

If $S_\Omega\leq\widetilde S_\Omega$ entrywise and both associated interior systems satisfy their spectral-radius conditions, then
\begin{equation}
\label{eq:boundary-response-monotonicity}
U_B(S_\Omega)\leq U_B(\widetilde S_\Omega)
\qquad\text{entrywise}.
\end{equation}
\end{proposition}

\begin{proof}
The Neumann representation gives
\[
U_B=\sum_{r=0}^{\infty}(S_\Omega W_{\Omega\Omega})^rS_\Omega W_{\Omega B},
\]
which is entrywise nonnegative. The unit-boundary, zero-interior-initial-condition experiment preserves the interval $[0,1]$, proving \eqref{eq:boundary-response-bounds} and \eqref{eq:boundary-response-row-sum}. The stated reachability condition supplies at least one strictly positive path contribution. Finally, every term in the displayed nonnegative path expansion is entrywise nondecreasing in the diagonal entries of $S_\Omega$, which proves \eqref{eq:boundary-response-monotonicity}.
\end{proof}

Interpreted as a network response matrix, the rows of $U_B$ summarize target-side receptivity to the declared boundary sources, whereas columns summarize the aggregate reach of individual fixed sources. The matrix is not an adjacency matrix, a transition matrix, or a topology-only kernel: it is a model-implied steady-response operator that combines multistep propagation, edge weights, boundary placement, and susceptibility.
When every vertex is scanned in turn as a source, the entries associated with target $b$ form column $b$ of the all-source response matrix $R$; their off-diagonal average is the reception score $Q_b$. Thus a high reception score records broad modeled responsiveness across the declared source set, rather than merely a large number or weight of incoming ties.

\subsection{Dynamic companion matrices and basic response properties}
\label{subsec:all-vertex-response-properties}

For completeness, the dynamic companion matrices are
\begin{equation}
\label{eq:dynamic-response-matrices}
\mathcal T_{ab}:=T_b^0(a),
\qquad
\mathcal E_{ab}:=E_{ab},
\qquad
\mathcal S_{ab}(\varepsilon):=T_b^{\mathrm{stab}}(a;\varepsilon).
\end{equation}
The empirical analyses in this article emphasize $R=\mathcal U^\infty$. The timing matrices in \eqref{eq:dynamic-response-matrices} should be reported only when their distinct substantive interpretation is used; otherwise, they may be supplied as supplementary material.

\begin{proposition}[Basic response-matrix properties]
\label{prop:response-matrix-properties}
For every $a,b\in V$ for which the source-indexed experiment is stable,
\begin{equation}
\label{eq:response-matrix-bounds}
0\leq R_{ab}\leq1,
\qquad
R_{aa}=1.
\end{equation}
If $b$ is not reachable from $a$ in the influence graph, then $R_{ab}=0$. The diagonal entries are excluded from source and target summary scores below because they reflect the imposed unit value of the source rather than propagation to another node.
\end{proposition}

\begin{proof}
The source-indexed dynamics preserve $[0,1]$, source $a$ remains fixed at one, and the absence of an influence path prevents the source value from reaching an unreachable target.
\end{proof}

\subsection{Response graph and dynamics-induced centralities}
\label{subsec:response-centralities}

The all-source matrix $R$ generally records positive multistep effects between nonadjacent vertices. It therefore induces a response graph rather than merely reweighting the original edge set. For a threshold $\tau\geq0$, define the off-diagonal response matrix
\begin{equation}
\label{eq:thresholded-response}
R^{(\tau)}_{ab}:=
\begin{cases}
R_{ab}, & a\neq b\text{ and }R_{ab}>\tau,\\
0, & \text{otherwise},
\end{cases}
\end{equation}
and the directed response graph
\begin{equation}
\label{eq:response-graph}
G_R^{(\tau)}=(V,E_R^{(\tau)}),
\qquad
E_R^{(\tau)}:=\{(a,b):R^{(\tau)}_{ab}>0\}.
\end{equation}
An edge $a\to b$ in $G_R^{(\tau)}$ means that the source-indexed FJ experiment at $a$ produces a positive steady response at $b$. The choice $\tau=0$ gives the mathematical response graph; positive thresholds used for visualization or numerical sparsification must be reported and subjected to sensitivity analysis.

\begin{definition}[Primary broadcasting and reception scores]
\label{def:primary-response-scores}
The \emph{broadcasting strength} and \emph{reception strength} of node $a$ are, respectively,
\begin{equation}
\label{eq:broadcasting-strength}
B_a:=\frac{1}{n-1}\sum_{b\neq a}R_{ab},
\qquad
Q_a:=\frac{1}{n-1}\sum_{b\neq a}R_{ba}.
\end{equation}
\end{definition}

$B_a$ summarizes the average off-diagonal steady response generated by source $a$, whereas $Q_a$ summarizes the average response received by target $a$ across the collection of source experiments. Thus $Q_a$ is scientifically useful as a model-defined measure of cross-source exposure or vulnerability: a high value indicates that the target's modeled steady state is broadly responsive to the admissible source interventions under the declared network, boundary, and susceptibility specification. It is not an intrinsic psychological trait and should not be read as an observational estimate of persuadability.

For a distance-based source score, let $d_R^{(\tau)}(a,b)$ be the unweighted directed shortest-path distance in $G_R^{(\tau)}$. The \emph{response-harmonic broadcasting score} is
\begin{equation}
\label{eq:response-harmonic-score}
H_a^{\mathrm{out}}:=\frac{1}{n-1}
\sum_{\substack{b\neq a\\d_R^{(\tau)}(a,b)<\infty}}
\frac{1}{d_R^{(\tau)}(a,b)}.
\end{equation}
Unlike a global-minimum-weight closeness formula, \eqref{eq:response-harmonic-score} is defined naturally on directed and possibly disconnected response graphs.

If a path-mediated score is required, define positive response lengths by
\begin{equation}
\label{eq:response-lengths}
\ell_{ab}:=\frac{1}{R^{(\tau)}_{ab}}
\qquad\text{for }(a,b)\in E_R^{(\tau)}.
\end{equation}
A directed weighted betweenness score can then be calculated from shortest paths under these lengths. Since this score measures the role of intermediaries in the response graph rather than direct source impact, it should be reported only when that distinction is substantively relevant.

Because rows of the response matrix index sources and columns index targets, orientation must be fixed before applying a recursive ranking routine. A source-oriented eigenvector score is the left Perron vector of $R^{(\tau)}$, whereas a target-oriented score is the corresponding right Perron vector. Ordinary PageRank on the directed response graph ranks targets by incoming response flow and is therefore reception-oriented; applying PageRank to the reversed response graph gives the corresponding broadcasting-oriented ranking. Any reported variant must state its orientation, damping or regularization, and dangling-node treatment. Implementation details and toy-network checks belong in the reproducibility materials rather than the main text.

For any node score $c$, the quantity
\begin{equation}
\label{eq:raw-centralization}
D_c(G_R^{(\tau)})
:=\sum_{a\in V}\left(\max_{b\in V}c_b-c_a\right)
\end{equation}
is an \emph{unnormalized dispersion statistic}. It should not be compared across response graphs of different sizes, thresholds, or score ranges without a defensible denominator. Cross-network comparisons should use a stated normalization or a scale-free inequality summary, such as a Gini coefficient, with its interpretation and benchmark made explicit.

\subsection{Relation to structural centrality and influence concepts}
\label{subsec:relation-structural-centrality}

The response-based quantities above are intended to complement, not replace, topology-only centralities. Centrality interpretation depends on the network process and question of interest \citep{Borgatti2005,BorgattiEverett2006,BoldiVigna2014}. Structural scores use the graph or its adjacency weights as input. In contrast, $R$ is generated by an FJ source-response experiment and therefore changes when susceptibility, boundary treatment, or the influence-update model changes. This distinction is especially important when nodes with similar topology have different propensities to retain their initial opinions.

\begin{table}[t]
\centering
\caption{Conceptual comparison between structural centralities and response-based quantities.}
\label{tab:structural-response-comparison}
\begin{tabular}{p{0.22\linewidth}p{0.31\linewidth}p{0.39\linewidth}}
\hline
\textbf{Structural concept} & \textbf{Primary input} & \textbf{Response-based distinction}\\
\hline
Degree & Adjacent links or their weights & $B_a$ aggregates the modeled steady response of all targets, including multistep effects.\\
Closeness or harmonic centrality & Directed shortest-path reachability & $H_a^{\mathrm{out}}$ measures reachability in the response graph induced by the FJ experiment.\\
Betweenness & Shortest paths in the observed graph & Response betweenness uses paths weighted by modeled steady responses rather than observed adjacency alone.\\
Eigenvector and PageRank & Recursive prestige or stationary flow in an observed graph & Their response-graph analogues require an explicit broadcasting or reception orientation.\\
Katz, communicability, and diffusion rankings & Attenuated walks, walk-based connectivity, or process-dependent flow & Proposition~\ref{prop:response-resolvent-factorization} gives the exact resolvent relation; source normalization, boundary conditioning, and heterogeneous susceptibility prevent reduction to a topology-only scalar resolvent in general.\\
\hline
\end{tabular}
\end{table}

Proposition~\ref{prop:response-resolvent-factorization} supplies the exact comparison with Katz-type resolvents. With strictly positive baseline susceptibilities, $R=D_G^{-1}G^\top$ for $G=(I-SW)^{-1}$; hence the response matrix is a source-normalized transposed FJ resolvent. Only when $S=\beta I$ and the diagonal of $G$ is constant does broadcasting become an affine transformation of the source-oriented Katz vector. Communicability uses the distinct exponential kernel $\exp(A)$ \citep{EstradaHatano2008}, and dynamic communicability and random-walk diffusion represent other process-dependent comparators \citep{GrindrodEtAl2011,MasudaPorterLambiotte2017}. HITS addresses a different hub--authority link-prestige question and is therefore outside the present empirical comparison \citep{Kleinberg1999}.

The distinction between broadcasting and reception also clarifies the relation to social-power and key-player questions. Broadcasting scores address which fixed source produces broad modeled response; reception scores address which target is responsive across sources; and timing measures address when a response first arrives or stabilizes. These questions may be aligned in some networks and distinct in others. Their empirical utility should be assessed by transparent comparisons with the structural baselines and by sensitivity analyses over susceptibility profiles, source values, thresholds, and directed-weight conventions.

\section{Computational Evaluation Across Empirical, Benchmark, and Synthetic Networks}
\label{sec:computational-evaluation}

\subsection{Objectives and evaluation design}
\label{subsec:evaluation-design}

This section evaluates the response-based broadcasting and reception scores from Section~\ref{subsec:response-centralities} across eight networks that differ in size, directionality, weights, community structure, and degree heterogeneity. The objective is to characterize how the model-implied response operator behaves under declared Friedkin--Johnsen (FJ) interventions and to compare its source- and target-oriented summaries with topology-only baselines.

For every network and every candidate source $a$, we set the source value to one, make $a$ fully stubborn, and set all other initial opinions to zero as in Section~\ref{subsec:source-response-scope}. The primary summaries are the off-diagonal broadcasting and reception strengths
\[
B_a=\frac{1}{n-1}\sum_{b\ne a}R_{ab},
\qquad
Q_a=\frac{1}{n-1}\sum_{b\ne a}R_{ba}.
\]
Thus, $B_a$ measures average modeled steady response generated by a source, whereas $Q_a$ measures average modeled steady response received by a target across source interventions.

We use twelve independently generated baseline susceptibility profiles for each network. In every profile, non-source susceptibilities are independent draws from a $\mathrm{Beta}(4,2)$ distribution rescaled to $[0.15,0.90]$. Because the receive-from matrix $W$ is nonnegative and row-substochastic and every non-source susceptibility is at most $0.90$, the bound $\rho(SW)\leq\lVert SW\rVert_\infty\leq\max_i s_i<1$ guarantees convergence in every source-indexed experiment. Reported node scores are means across the twelve profiles.

Each input edge is stored in causal communication direction $u\to v$, meaning that $u$ can influence $v$. The FJ receive-from matrix is formed with $W_{vu}>0$ when the edge $u\to v$ is present. Every row with positive raw incoming weight is normalized to sum to one. A vertex with no receive-from neighbours is retained, is assigned an all-zero row, and receives no added self-loop; hence the resulting $W$ is nonnegative and row-substochastic rather than necessarily row-stochastic. For a non-source zero-row vertex, the social-input term is zero and, under the zero interior-initial-condition response experiment, every off-diagonal response entry targeting that vertex is zero; if the vertex is selected as a source, it is instead made a unit stubborn boundary and can still propagate through any outgoing arcs. Undirected observed or benchmark networks are represented by reciprocal arcs. Because all baseline susceptibility values are strictly positive before a source is made stubborn, Proposition~\ref{prop:response-resolvent-factorization} gives $R_{ab}=G_{ba}/G_{aa}$ with baseline resolvent $G=(I-SW)^{-1}$. We form $G$ by dense linear algebra for the six networks with at most 986 processed vertices. For the Conover and Congress Twitter cases, exact sparse full-system solves together with exact diagonal-Green calculations by strongly connected diagonal blocks recover the same score summaries without materializing a dense $R$. Solver residuals, input hashes, seeds, and package versions are recorded in the reproducibility materials.

\subsection{Benchmark networks and preprocessing}
\label{subsec:network-portfolio}

The email-Eu-core case is the only directed observed network in the benchmark collection and is therefore important for checking the source-to-target convention beyond reciprocal benchmark graphs. The karate-club network is retained because its size makes response maps and rankings transparent. The Les Mis\'{e}rables network is included as a weighted computational benchmark, not as evidence of real interpersonal influence. The synthetic networks are not intended to mimic a single real social system; instead, they provide controlled settings in which directionality, modularity, and degree heterogeneity can be separated.

Table~\ref{tab:eight-network-portfolio} summarizes the eight-network benchmarks. Three observed cases preserve transparent reference settings; three synthetic directed families isolate null connectivity, community bottlenecks, and asymmetric hub structure; and two Twitter interaction networks add substantively interpreted, directed weighted communication cases. The figures use the processed arc counts reported in the table, so that reciprocal representations of undirected networks remain explicit.

\begin{table}[t]
\centering
\caption{Eight-network evaluation benchmarks. The arc count is the number of directed arcs used by the FJ analysis after the stated preprocessing. The two Twitter cases are directed weighted communication/interaction networks; neither is treated as a direct observation of interpersonal opinion influence.}
\label{tab:eight-network-portfolio}
\scriptsize
\begin{tabular}{p{1.80cm}p{1.35cm}rrp{4.60cm}}
\hline
Network & Role & $n$ & arcs & Construction and interpretation  \\
\hline
Karate club & Observed benchmark & 34 & 156 & Friendship ties represented as reciprocal unit-weight arcs \citep{Zachary1977}. \\
email-Eu-core & Observed directed & 986 & 24,929 & Emails represented as directed arcs; self-loops removed; no symmetrization \citep{YinEtAl2017,LeskovecKleinbergFaloutsos2007}. \\
Les Mis\'erables & Weighted benchmark & 77 & 508 & Weighted character co-appearance ties represented as reciprocal arcs \citep{Knuth1993}. \\
Directed Erd\H{o}s--R\'enyi & Directed null & 900 & 8,917 & One seeded directed $G(n,p)$ draw, expected out-degree 10 \citep{ErdosRenyi1959}. \\
Directed DC-SBM & Directed community & 900 & 9,172 & Four-block directed degree-corrected construction \citep{KarrerNewman2011}. \\
Directed preferential attachment & Directed hub & 900 & 3,172 & Seeded directed preferential-attachment construction \citep{BarabasiAlbert1999}. \\
Conover Twitter retweet & Observed retweet & 18,470 & 48,365 & Repeated 2010 political retweets aggregated to source-to-retweeter weights; no symmetrization \citep{Conover2011}. \\
Congress Twitter & Observed interaction & 475 & 13,289 & Supplied 2022 pairwise Twitter interaction/transmission-probability weights retained \citep{Fink2023CongressTwitter}. \\
\hline
\end{tabular}
\end{table}

The Conover retweet graph represents political retweet events in the six weeks preceding the 2010 U.S. congressional midterm elections. An arc from $u$ to $v$ records that content originating at $u$ was retweeted by $v$, and repeated events are aggregated as arc weights. The Congress graph retains the supplied directed pairwise interaction/transmission-probability weights among members of the 117th U.S. Congress. These two Twitter cases enlarge the substantive scope of the evaluation networks, but their interaction topologies are not treated as direct observations of interpersonal opinion change. Named accounts, where used in descriptive rank contrasts, are public-network labels only and are not evidence about personal beliefs, persuasion, or actual opinion change. The synthetic networks remain controlled constructions rather than empirical ground truth.

\subsection{Results}
\label{subsec:evaluation-results}

Figures~\ref{fig:broadcasting-outstrength} and \ref{fig:reception-instrength} retain node-level scatter plots for the six original cases. Across those cases, broadcasting is strongly positively associated with out-strength in five networks, whereas the directed preferential-attachment construction has a substantially lower mean Spearman correlation of $0.438$. The complete eight-network comparison below moves beyond correlations by also reporting top-ten overlap, a model-based source-selection diagnostic, and maximum rank displacement.

\begin{figure}[t]
\centering
\includegraphics[width=\textwidth]{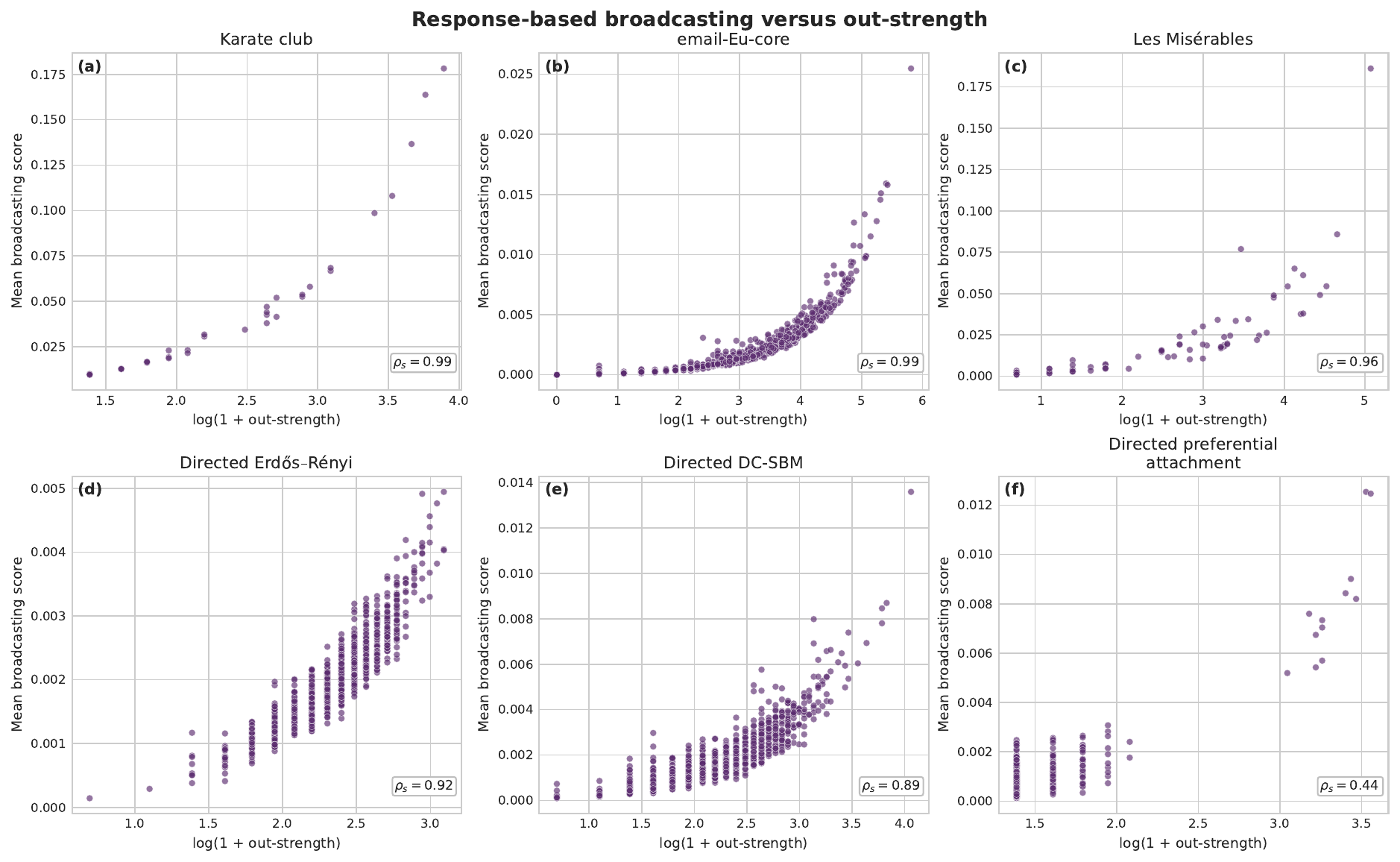}
\caption{Mean FJ broadcasting strength versus source out-strength in the six original networks. Each point is a node. Scores are averaged over twelve susceptibility profiles; the displayed $\rho_s$ is the Spearman rank correlation. The complete eight-network comparison is reported in Table~\ref{tab:eight-network-response-comparison} and Figure~\ref{fig:eight-network-response-divergence}.}
\label{fig:broadcasting-outstrength}
\end{figure}

Figure~\ref{fig:reception-instrength} retains the corresponding six-case view of reception. Its small or negative correlations in the reciprocal-arc networks and in email-Eu-core illustrate that reception is a response-defined target property rather than a relabeling of local in-strength. The full eight-network results below show that this distinction persists in the added Twitter cases.

\begin{figure}[t]
\centering
\includegraphics[width=\textwidth]{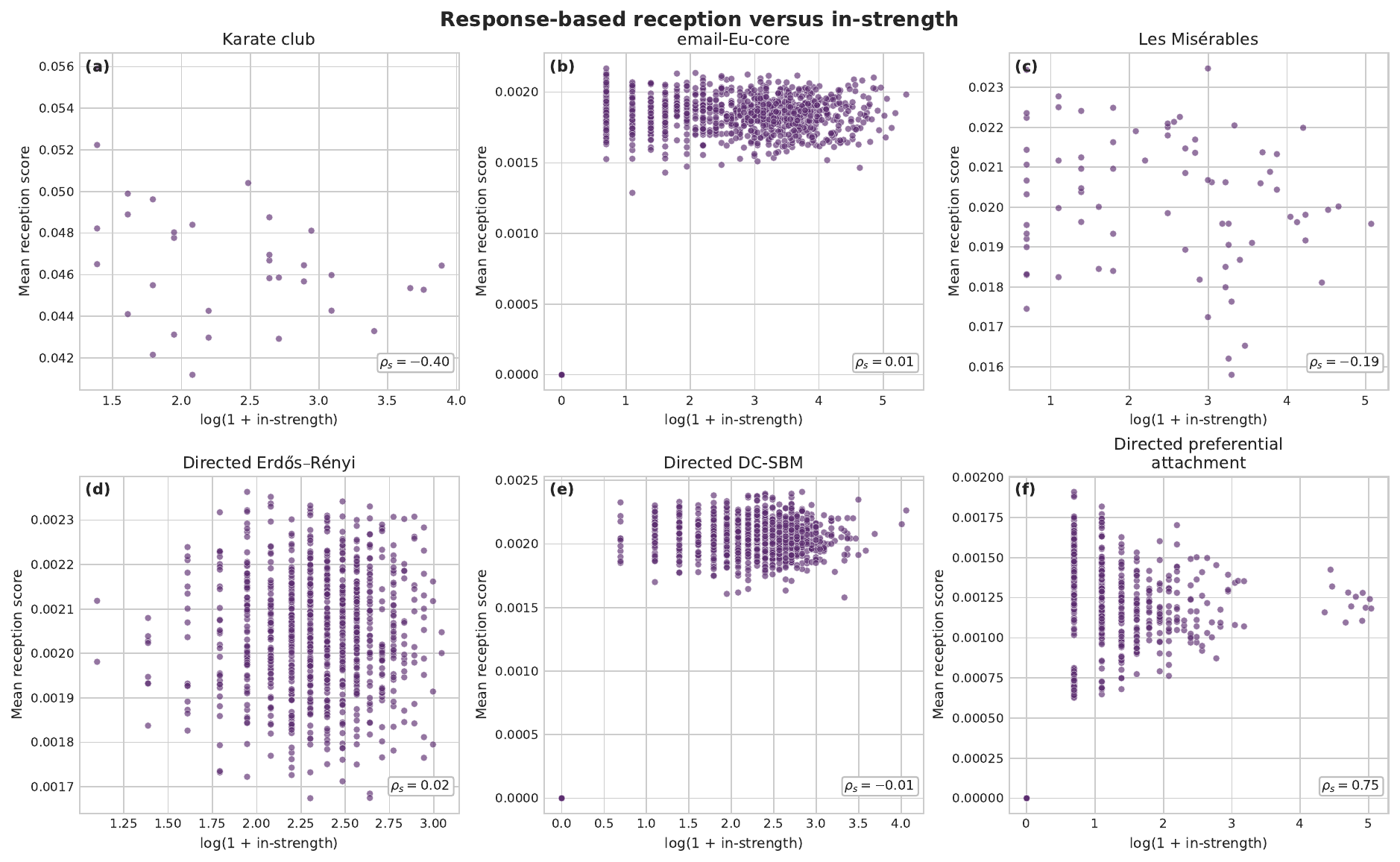}
\caption{Mean FJ reception strength versus target in-strength in the six original networks. Each point is a node. Scores are averaged over twelve susceptibility profiles; the displayed $\rho_s$ is the Spearman rank correlation. The complete eight-network comparison is reported in Table~\ref{tab:eight-network-response-comparison} and Figure~\ref{fig:eight-network-response-divergence}.}
\label{fig:reception-instrength}
\end{figure}

The eight-network comparison in Table~\ref{tab:eight-network-response-comparison} demonstrates both agreement and nontrivial divergence. Broadcasting--out-strength rank correlations range from $0.441$ to $0.989$, but top-ten overlap ranges from $0.533$ to $0.917$ and the model-based top-ten selection efficiency ranges from $0.908$ to $0.995$. The weakest global association occurs in the directed preferential-attachment case ($\rho=0.441$); the Congress Twitter case also has a lower association ($\rho=0.829$) and mean top-ten overlap of $0.717$. Even in the Conover retweet graph, where the global correlation is $0.988$, the mean maximum rank displacement is 5,965 positions. Thus a high global correlation need not reproduce the complete source ranking.

The node-level contrasts make this distinction concrete. In the Congress Twitter network, \texttt{RepBobGood} has the highest weighted out-strength (rank $1$) but ranks $80$ by broadcasting; its weighted out-strength is $0.63055$, whereas its mean broadcasting score is $0.00108$. Conversely, \texttt{RepBost} ranks $468$ by weighted in-strength ($0.00478$) but ranks first by reception ($0.00440$). These contrasts show that local incoming or outgoing strength need not reproduce the multistep, source-conditioned FJ response ordering. Broadcasting is closely aligned with the deliberately comparable Katz baseline in these simulations ($\rho=0.965$--$0.999$), as Proposition~\ref{prop:response-resolvent-factorization} helps explain; reception provides the sharper empirical distinction, with near-zero top-ten overlap with weighted in-strength in several networks. The selection-efficiency column remains a model-based comparison under the same FJ objective.

\begin{table}[t]
\centering
\caption{Eight-network comparison of model-based FJ scores with strength-based benchmarks. Values are means across 12 independently generated $\mathrm{Beta}(4,2)$ susceptibility profiles rescaled to $[0.15,0.90]$. `Top-10' is the mean fraction of nodes shared by the FJ broadcasting top ten and the weighted out-strength top ten; `$Q$ top-10' is the corresponding mean fraction shared by the FJ reception top ten and the weighted in-strength top ten. `Efficiency' is the mean FJ broadcasting among the top ten nodes selected by weighted out-strength, divided by the corresponding mean for the FJ top ten; it is a model-based selection diagnostic, not external causal validation.}
\label{tab:eight-network-response-comparison}
\scriptsize
\resizebox{\textwidth}{!}{%
\begin{tabular}{lrrrrrr}
\hline
Network & $\rho(B,d^{out})$ & Top-10 & Efficiency & $\rho(B,\mathrm{Katz})$ & $\rho(Q,d^{in})$ & $Q$ top-10  \\
\hline
Karate club & 0.974 & 0.917 & 0.995 & 0.975 & -0.077 & 0.300 \\
email-Eu-core & 0.989 & 0.700 & 0.950 & 0.996 & 0.053 & 0.000 \\
Les Mis\'erables & 0.950 & 0.792 & 0.931 & 0.989 & 0.012 & 0.075 \\
Directed Erd\H{o}s--R\'enyi & 0.907 & 0.533 & 0.930 & 0.974 & 0.005 & 0.008 \\
Directed DC-SBM & 0.886 & 0.583 & 0.908 & 0.985 & 0.014 & 0.000 \\
Directed preferential attachment & 0.441 & 0.883 & 0.974 & 0.965 & 0.772 & 0.000 \\
Conover Twitter retweet & 0.988 & 0.792 & 0.986 & 0.999 & 0.587 & 0.000 \\
Congress Twitter & 0.829 & 0.717 & 0.909 & 0.989 & 0.030 & 0.008 \\
\hline
\end{tabular}%
}
\end{table}

\begin{figure}[t]
\centering
\includegraphics[width=\textwidth]{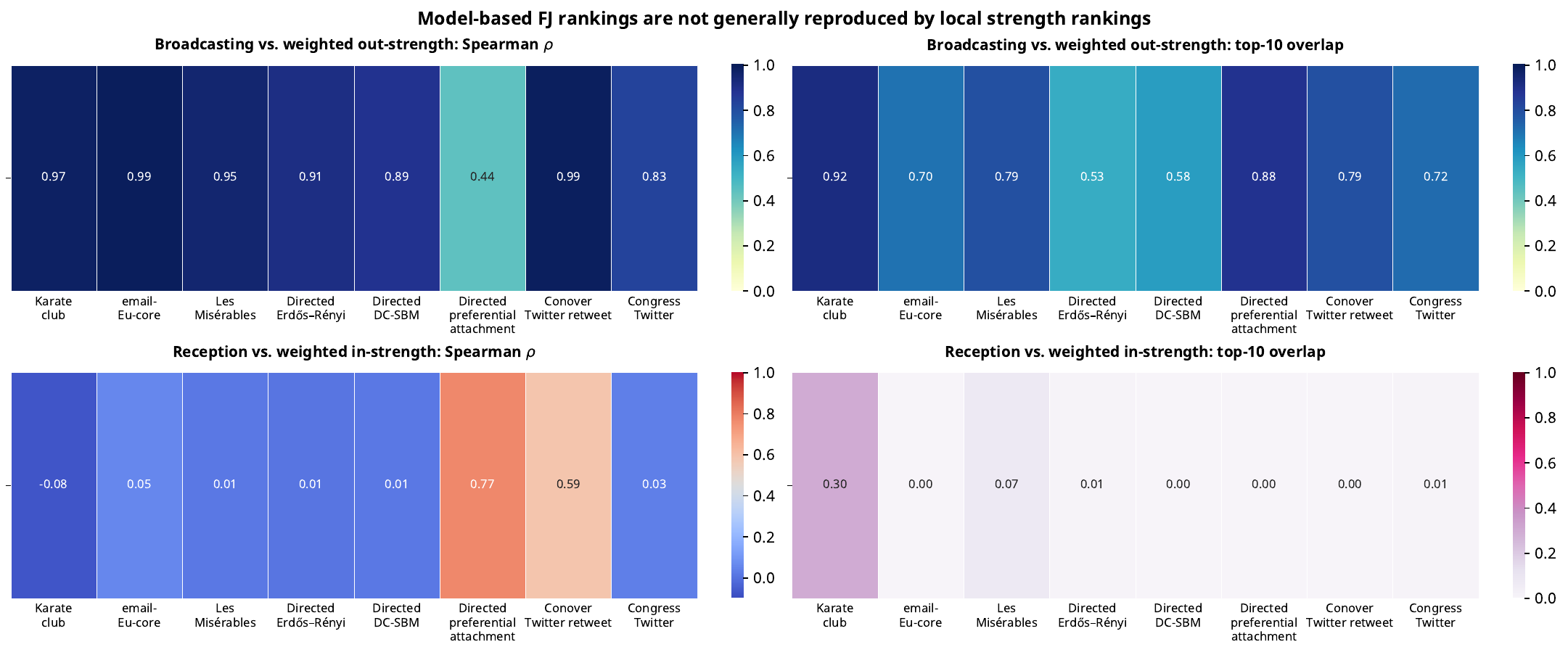}
\caption{Eight-network comparison of FJ response scores with local strength rankings. The upper panels show mean Spearman correlations and mean top-ten overlap between broadcasting and weighted out-strength; the lower panels show the corresponding quantities between reception and weighted in-strength. All values average twelve independently generated $\mathrm{Beta}(4,2)$ susceptibility profiles rescaled to $[0.15,0.90]$.}
\label{fig:eight-network-response-divergence}
\end{figure}

The detailed six-original-network score-distribution, robustness, and computational-performance results are moved to Supplementary Material Sections~S1--S2. The main text retains the complete eight-network comparison in Table~\ref{tab:eight-network-response-comparison} and the matched heterogeneous-versus-homogeneous susceptibility analysis below because these results directly quantify cross-network rank divergence and a decision-relevant change in the selected top broadcaster.

To demonstrate a decision-relevant consequence of susceptibility heterogeneity rather than a change in its mean level, we compare every heterogeneous profile with a homogeneous profile having the same mean susceptibility across the complete eight-network benchmarks. The broad broadcasting order is often stable (mean Spearman correlations $0.976$--$1.000$), but the selected leading source can change. In baseline profile 2 of the Conover retweet network, whose mean susceptibility is $0.6494$, the heterogeneous profile ranks source $2624$ first whereas the matched homogeneous profile ranks source $10391$ first. This is not merely a local rank shift: an analyst using the homogeneous approximation would select a different first broadcaster from an analyst using the profile-aware response calculation, despite identical mean susceptibility and a high global rank correlation. In the same comparison, node $2018$ moves from rank $1{,}530$ under homogeneous susceptibility to rank $2{,}870$ under heterogeneous susceptibility. The change arises because heterogeneous susceptibility reweights downstream continuation and absorption along different response paths even though the mean susceptibility is held fixed. Across the twelve Conover profiles, the top source changes in four cases and the largest individual displacement is $1{,}624$ positions. Heterogeneity can therefore change the practical source-selection decision even when broad rank correlation remains high.

\begin{table}[t]
\centering
\caption{Effect of heterogeneous susceptibility on FJ broadcasting rankings. Each heterogeneous profile is compared with a homogeneous profile having the same mean susceptibility. Values summarize 12 independently generated baseline profiles.}
\label{tab:eight-network-heterogeneity}
\scriptsize
\begin{tabular}{lrrrr}
\hline
Network & Mean $\rho$ & Mean top-10 overlap & Changed top source & Largest rank shift  \\
\hline
Karate club & 0.976 & 0.917 & 2/12 & 8 \\
email-Eu-core & 0.999 & 0.950 & 0/12 & 78 \\
Les Mis\'erables & 0.990 & 0.983 & 0/12 & 13 \\
Directed Erd\H{o}s--R\'enyi & 0.977 & 0.700 & 6/12 & 259 \\
Directed DC-SBM & 0.989 & 0.850 & 0/12 & 225 \\
Directed preferential attachment & 0.980 & 0.967 & 5/12 & 284 \\
Conover Twitter retweet & 1.000 & 0.967 & 4/12 & 1624 \\
Congress Twitter & 0.994 & 0.933 & 0/12 & 101 \\
\hline
\end{tabular}
\end{table}

\begin{figure}[t]
\centering
\includegraphics[width=\textwidth]{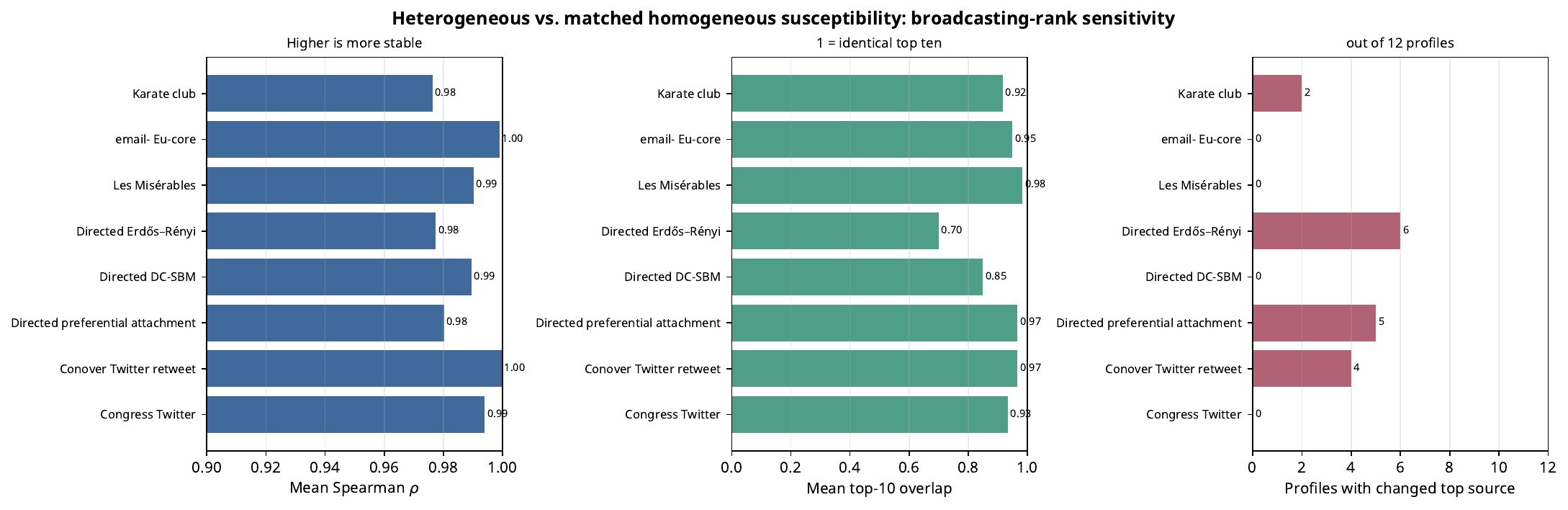}
\caption{Broadcasting-rank sensitivity to heterogeneous susceptibility across eight networks. Each heterogeneous profile is compared with a homogeneous profile having the same mean susceptibility. The panels report mean Spearman correlation, mean top-ten overlap, and the number of the twelve profiles in which the top broadcasting source changes.}
\label{fig:eight-network-heterogeneous-susceptibility}
\end{figure}

\subsection{Scope of the evaluation}

The present evaluation is a model-based computational study. The observed networks encode friendship ties, institutional email contacts, retweet activity, or congressional Twitter interactions; Les Mis\'{e}rables is a fictional co-appearance benchmark; and the three synthetic networks are generated examples. The findings therefore characterize the numerical behavior and comparative rankings of the proposed FJ response measures under the declared source interventions, susceptibility profiles, and edge conventions. Future work can assess these quantities against task-specific outcomes, temporal observations, or experimentally identified influence processes.

\section{Discussion, limitations, and use cases}
\label{sec:discussion}

\subsection{What the response scores measure}

The response construction distinguishes three quantities that should not be conflated. First, the broadcasting score $B_a$ is source-oriented: under the declared FJ source intervention, it summarizes the average steady response generated by source $a$ among the remaining nodes. Second, the reception score $Q_b$ is target-oriented: it is the average off-diagonal entry in column $b$ of the all-source response matrix and therefore summarizes how broadly target $b$ responds when the other nodes are considered in turn as sources. Third, structural centralities such as degree, betweenness, and PageRank are topology-based summaries. They may be associated with the response scores in a particular network and susceptibility profile, but they do not encode the FJ assumptions about stubbornness, heterogeneous susceptibility, or the chosen source intervention.

For a practitioner, a high $Q_b$ flags a target at which many admissible source interventions converge in their predicted steady-state effect. Such a node can be a priority location for monitoring information exposure, designing safeguards against unwanted persuasion, or assessing information accessibility, depending on the stated application. Mathematically, a high reception score records the joint effect of the target's position in the influence network and the declared susceptibility field: it identifies where cross-source responsiveness is concentrated, rather than merely where incoming ties are numerous. Thus reception is useful when the question is which targets are broadly responsive across alternative sources, while broadcasting is useful when the question is which sources generate broad response. Both are conditional, model-implied summaries rather than universal replacements for structural centrality.

\begin{figure}[t]
\centering
\includegraphics[width=0.96\textwidth]{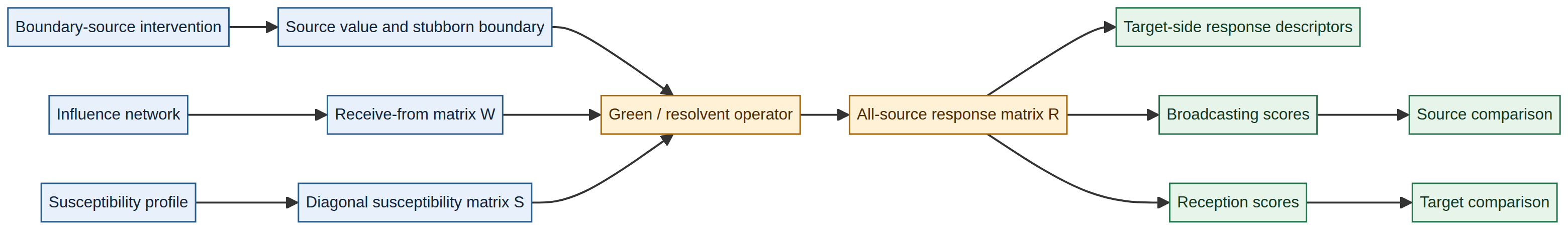}
\caption{Conceptual workflow for the response-based measures. Boundary-source interventions, the influence network, and susceptibility parameters determine the Green/resolvent operator and the all-source response matrix. Target-side response descriptors and source- or target-oriented centralities are then derived from that response matrix.}
\label{fig:fj-response-workflow}
\end{figure}

\subsection{Use cases and intervention interpretation}

The framework is useful for scenario analysis when an analyst wishes to compare hypothetical source interventions under an explicitly declared influence network and susceptibility profile. For example, broadcasting can compare candidate communication sources under one common FJ specification, while reception can identify targets whose modeled steady states are broadly responsive across the admissible source set. The resulting rankings support transparent sensitivity analysis over the assumptions that define the scenario.

The general formulation accommodates directed weighted influence networks and heterogeneous susceptibility profiles. The email-Eu-core, Conover retweet, and Congress Twitter cases illustrate the stated source-to-target convention on observed directed interaction networks; the synthetic families separately isolate null connectivity, community structure, and asymmetric hub formation. In each case, the response scores are calculated only after the edge convention, susceptibility profile, and source intervention have been specified.

\subsection{Limitations}

Several limitations delimit the present findings. The analysis uses linear FJ dynamics, a fixed influence network, and susceptibility values specified exogenously rather than estimated from observed opinion trajectories. The response scores may therefore change when the influence matrix, susceptibility distribution, source value, response threshold, or aggregation rule changes. Moreover, the observed networks are proxies for social or communication contact rather than direct observations of interpersonal opinion change. The computations consequently characterize the stated model and do not identify causal influence or validate a behavioral mechanism.

There are also computational and empirical-scope limitations. An all-source analysis can be costly on large or dynamic networks because it requires many source-specific response calculations and, in a multi-profile design, repeated susceptibility draws. The closed-form identity and sparse implementation make the present analyses feasible, but larger-scale use requires dedicated numerical study. The observed cases demonstrate the method on transparent benchmark and communication networks; they do not establish performance for a particular intervention task.

\subsection{Extensions}

Natural extensions include empirical estimation of susceptibility and influence weights from longitudinal opinion data, source-selection objectives tied to observed outcomes, and analyses of time-varying or multilayer networks. Other directions include signed influence, multidimensional opinions, adaptive susceptibility, and nonlinear opinion-update mechanisms. Each extension would require a new model formulation and validation strategy; it should not be inferred automatically from the present linear framework.

\section{Conclusion}
\label{sec:conclusion}

This article formulated Friedkin--Johnsen dynamics with fully stubborn sources as a discrete network boundary-value problem. The resulting Green-operator representation yields source-conditioned steady-state responses and supports target-side influenceability descriptors, broadcasting scores, and reception scores. The response--resolvent factorization clarifies the relationship to resolvent-based centralities: the all-source response is source-normalized, boundary-conditioned, and susceptibility-dependent in general, whereas a homogeneous constant-diagonal specialization yields an affine Katz-type broadcasting relation.

The eight-network computational evaluation shows both agreement and divergence relative to topology-only baselines under declared FJ interventions. Broadcasting often aligns with source-oriented Katz/resolvent rankings, but not always with local out-strength. Top-ten agreement, rank displacement, and the model-based selection diagnostic show that local strength does not invariably recover source ordering. Reception is more clearly distinct from target in-strength in several settings, while matched homogeneous-versus-heterogeneous comparisons show that susceptibility can change top-source identification even when broad rank correlation remains high.

\clearpage
\section*{Supplementary Material}
\label{sec:supplementary-material}

\subsection*{S1. Detailed six-original-network score distributions}

The score distributions in Figure~\ref{fig:response-score-distributions} further distinguish the six original network families. The directed preferential-attachment and directed DC-SBM networks display a small number of high-broadcasting nodes relative to their typical nodes, whereas the directed Erd\H{o}s--R\'{e}nyi network has a more compressed distribution. The observed email network also contains a pronounced source-response tail. These descriptive patterns are consistent with the intended role of the synthetic constructions: the three families expose null, modular, and hub-dominated propagation settings without treating any one as an empirical ground truth.

\begin{figure}[p]
\centering
\includegraphics[width=\textwidth]{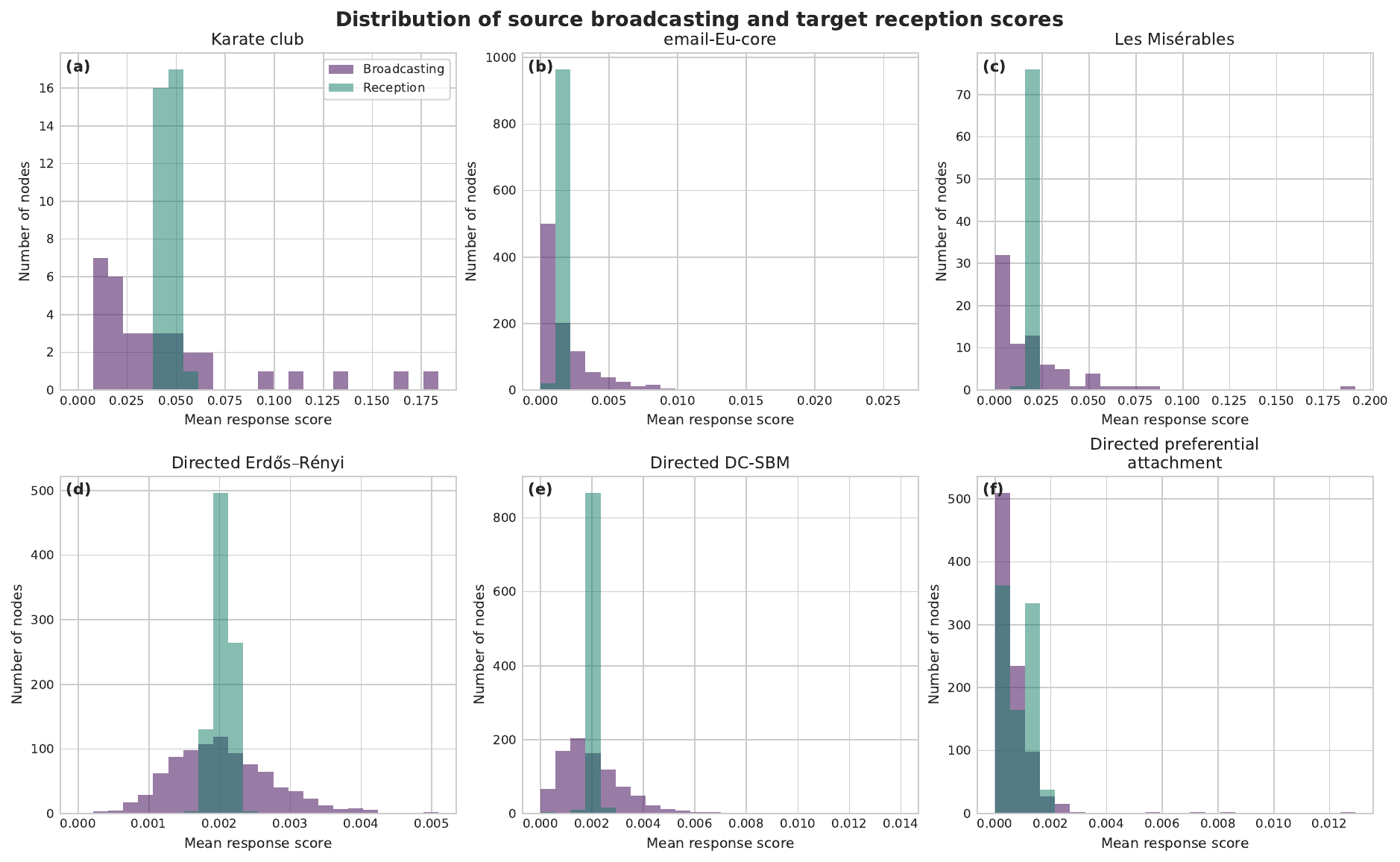}
\caption{Distribution of mean broadcasting and reception scores across the six original networks. The diagonal source response is excluded from both scores.}
\label{fig:response-score-distributions}
\end{figure}

\subsection*{S2. Robustness analyses for the six original networks}

For the six original small- and medium-network cases, we assessed ranking robustness by varying the non-source susceptibility distribution and the placement of neutral background boundary agents. The baseline profiles use $\mathrm{Beta}(4,2)$ draws rescaled to $[0.15,0.90]$. The two alternatives use $\mathrm{Beta}(2,4)$ and $\mathrm{Beta}(6,2)$ draws on the same range. For each distribution, twelve independently generated profiles were evaluated. In a separate boundary-placement analysis, zero-valued background agents were fixed at the lowest or highest weighted out-strength vertices, at one or five percent of the network, while the selected source remained a unit fixed source in its own experiment. All score summaries use the unthresholded response matrix ($\tau=0$); no edge perturbation was applied.

Table~\ref{tab:reviewer2-susceptibility-sensitivity} reports Spearman rank correlations and top-ten overlap with the baseline score rankings across the six original cases. Broadcasting rankings are highly stable under both alternative distributions ($\rho_B=0.979$--$1.000$; top-ten overlap $0.70$--$1.00$). Reception rankings vary more strongly with the susceptibility specification ($\rho_Q=-0.086$--$0.983$; top-ten overlap $0.00$--$0.50$). Table~\ref{tab:reviewer2-boundary-sensitivity} and Figure~\ref{fig:reviewer2-boundary-stability} show the corresponding neutral-boundary results. Broadcasting remains comparatively stable across the tested placements ($\rho_B=0.925$--$1.000$; top-ten overlap $0.50$--$1.00$), whereas reception has a wider response range ($\rho_Q=0.375$--$0.999$; top-ten overlap $0.30$--$1.00$). Thus, the source-oriented broadcasting result is robust across these specification checks, while target-oriented reception rankings should be interpreted as conditional on the susceptibility and boundary treatment.

\begin{table*}[p]
\centering
\caption{Ranking stability under alternative susceptibility distributions across the six original networks. Each reported score is a mean across 12 independently generated profiles. $\rho_B$ and $\rho_Q$ are Spearman correlations with the baseline $\mathrm{Beta}(4,2)$ ranking; overlap is the fraction of common nodes among the top ten.}
\label{tab:reviewer2-susceptibility-sensitivity}
\scriptsize
\resizebox{\textwidth}{!}{%
\begin{tabular}{llrrrr}
\hline
Network & Susceptibility condition & $\rho_B$ & Top-10 $B$ overlap & $\rho_Q$ & Top-10 $Q$ overlap  \\
\hline
Karate club & Lower $\mathrm{Beta}(2,4)$ & 0.979 & 0.80 & -0.086 & 0.30  \\
Karate club & Higher $\mathrm{Beta}(6,2)$ & 0.995 & 1.00 & 0.448 & 0.40  \\
email-Eu-core & Lower $\mathrm{Beta}(2,4)$ & 0.997 & 0.90 & 0.081 & 0.00  \\
email-Eu-core & Higher $\mathrm{Beta}(6,2)$ & 1.000 & 0.90 & 0.087 & 0.00  \\
Les Mis\'erables & Lower $\mathrm{Beta}(2,4)$ & 0.987 & 1.00 & 0.092 & 0.20  \\
Les Mis\'erables & Higher $\mathrm{Beta}(6,2)$ & 0.995 & 1.00 & 0.127 & 0.20  \\
Directed Erd\H{o}s--R\'enyi & Lower $\mathrm{Beta}(2,4)$ & 0.989 & 0.70 & 0.004 & 0.00  \\
Directed Erd\H{o}s--R\'enyi & Higher $\mathrm{Beta}(6,2)$ & 0.996 & 0.80 & -0.003 & 0.00  \\
Directed DC-SBM & Lower $\mathrm{Beta}(2,4)$ & 0.991 & 0.90 & 0.110 & 0.10  \\
Directed DC-SBM & Higher $\mathrm{Beta}(6,2)$ & 0.998 & 1.00 & 0.097 & 0.00  \\
Directed preferential attachment & Lower $\mathrm{Beta}(2,4)$ & 0.990 & 1.00 & 0.944 & 0.00  \\
Directed preferential attachment & Higher $\mathrm{Beta}(6,2)$ & 0.997 & 1.00 & 0.983 & 0.50  \\
\hline
\end{tabular}%
}
\end{table*}

\begin{table*}[p]
\centering
\caption{Sensitivity to neutral background boundary agents across the six original networks. Background nodes have fixed zero opinions; source nodes are still assigned a unit fixed opinion in their source-specific experiment. Scores are compared with the source-only baseline using the same 12 baseline susceptibility profiles.}
\label{tab:reviewer2-boundary-sensitivity}
\scriptsize
\resizebox{\textwidth}{!}{%
\begin{tabular}{llrrrr}
\hline
Network & Background boundary placement & $\rho_B$ & Top-10 $B$ overlap & $\rho_Q$ & Top-10 $Q$ overlap  \\
\hline
Karate club & Lowest out-strength, 1\% ($k=1$) & 1.000 & 1.00 & 0.976 & 0.90  \\
Karate club & Highest out-strength, 1\% ($k=1$) & 0.947 & 0.80 & 0.597 & 0.80  \\
Karate club & Lowest out-strength, 5\% ($k=2$) & 0.999 & 1.00 & 0.888 & 0.90  \\
Karate club & Highest out-strength, 5\% ($k=2$) & 0.978 & 0.80 & 0.375 & 0.50  \\
email-Eu-core & Lowest out-strength, 1\% ($k=10$) & 1.000 & 1.00 & 0.989 & 1.00  \\
email-Eu-core & Highest out-strength, 1\% ($k=10$) & 0.998 & 1.00 & 0.727 & 0.40  \\
email-Eu-core & Lowest out-strength, 5\% ($k=49$) & 0.999 & 1.00 & 0.888 & 0.80  \\
email-Eu-core & Highest out-strength, 5\% ($k=49$) & 0.988 & 0.90 & 0.436 & 0.30  \\
Les Mis\'erables & Lowest out-strength, 5\% ($k=4$) & 0.996 & 1.00 & 0.943 & 1.00  \\
Les Mis\'erables & Highest out-strength, 5\% ($k=4$) & 0.968 & 0.90 & 0.398 & 0.30  \\
Les Mis\'erables & Lowest out-strength, 1\% ($k=1$) & 1.000 & 1.00 & 0.999 & 1.00  \\
Les Mis\'erables & Highest out-strength, 1\% ($k=1$) & 0.983 & 0.80 & 0.624 & 0.60  \\
Directed Erd\H{o}s--R\'enyi & Lowest out-strength, 1\% ($k=9$) & 0.999 & 1.00 & 0.962 & 0.90  \\
Directed Erd\H{o}s--R\'enyi & Highest out-strength, 1\% ($k=9$) & 0.983 & 0.80 & 0.864 & 0.80  \\
Directed Erd\H{o}s--R\'enyi & Lowest out-strength, 5\% ($k=45$) & 0.988 & 0.90 & 0.795 & 0.70  \\
Directed Erd\H{o}s--R\'enyi & Highest out-strength, 5\% ($k=45$) & 0.938 & 0.50 & 0.612 & 0.40  \\
Directed DC-SBM & Lowest out-strength, 1\% ($k=9$) & 1.000 & 1.00 & 0.969 & 1.00  \\
Directed DC-SBM & Highest out-strength, 1\% ($k=9$) & 0.986 & 0.90 & 0.840 & 0.40  \\
Directed DC-SBM & Lowest out-strength, 5\% ($k=45$) & 0.996 & 0.90 & 0.851 & 0.90  \\
Directed DC-SBM & Highest out-strength, 5\% ($k=45$) & 0.947 & 0.70 & 0.592 & 0.30  \\
Directed preferential attachment & Lowest out-strength, 5\% ($k=45$) & 0.937 & 1.00 & 0.935 & 0.80  \\
Directed preferential attachment & Lowest out-strength, 1\% ($k=9$) & 0.991 & 1.00 & 0.986 & 1.00  \\
Directed preferential attachment & Highest out-strength, 1\% ($k=9$) & 0.992 & 1.00 & 0.900 & 0.60  \\
Directed preferential attachment & Highest out-strength, 5\% ($k=45$) & 0.925 & 1.00 & 0.809 & 0.30  \\
\hline
\end{tabular}%
}
\end{table*}

\begin{figure*}[p]
\centering
\includegraphics[width=0.98\textwidth]{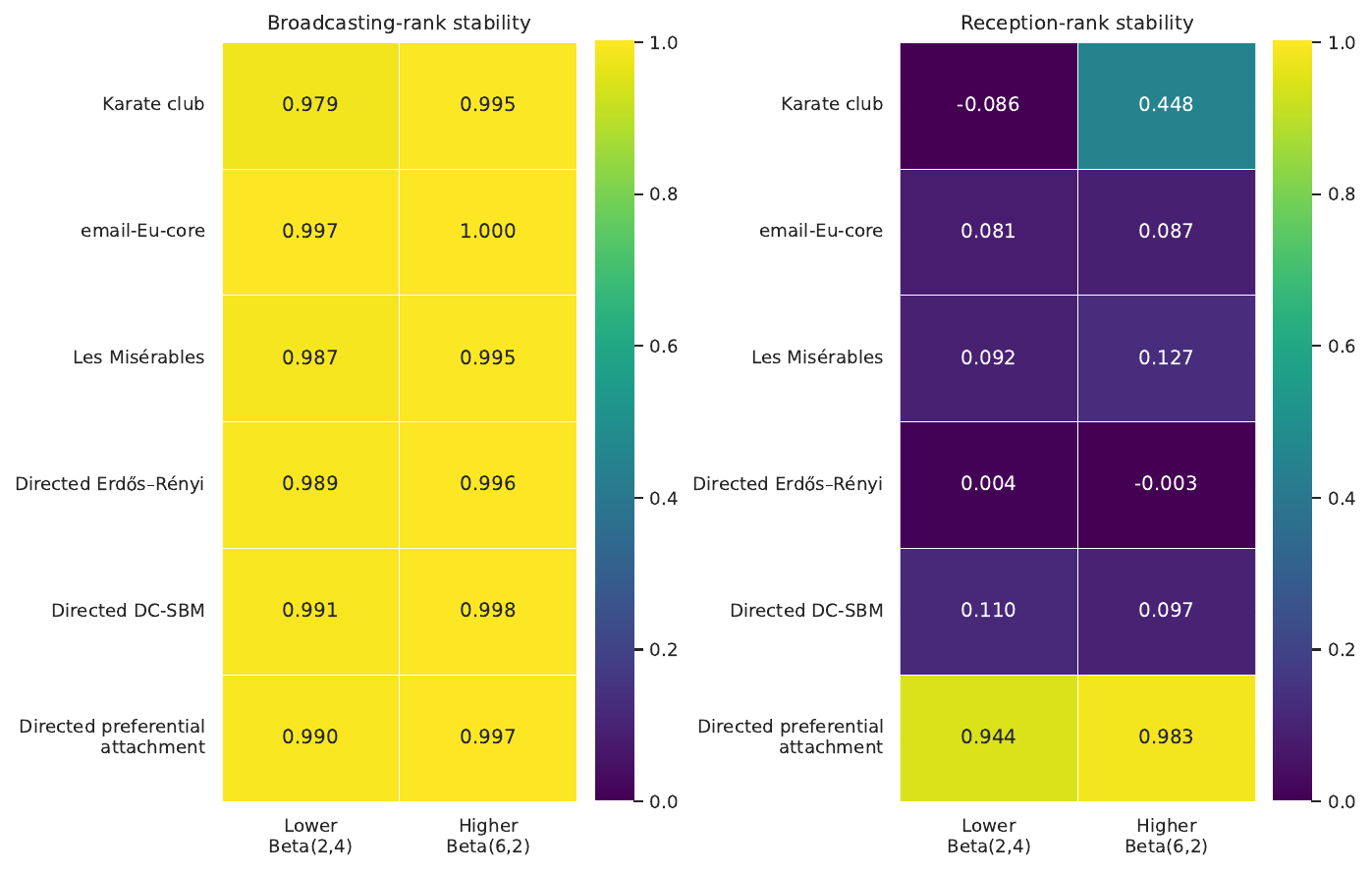}
\caption{Sensitivity of the mean broadcasting and reception rankings to alternative susceptibility distributions across the six original networks. Each annotation is the Spearman correlation with the baseline $\mathrm{Beta}(4,2)$ ranking after twelve independently generated profiles per condition.}
\label{fig:reviewer2-susceptibility-stability}
\end{figure*}

\begin{figure*}[p]
\centering
\includegraphics[width=0.98\textwidth]{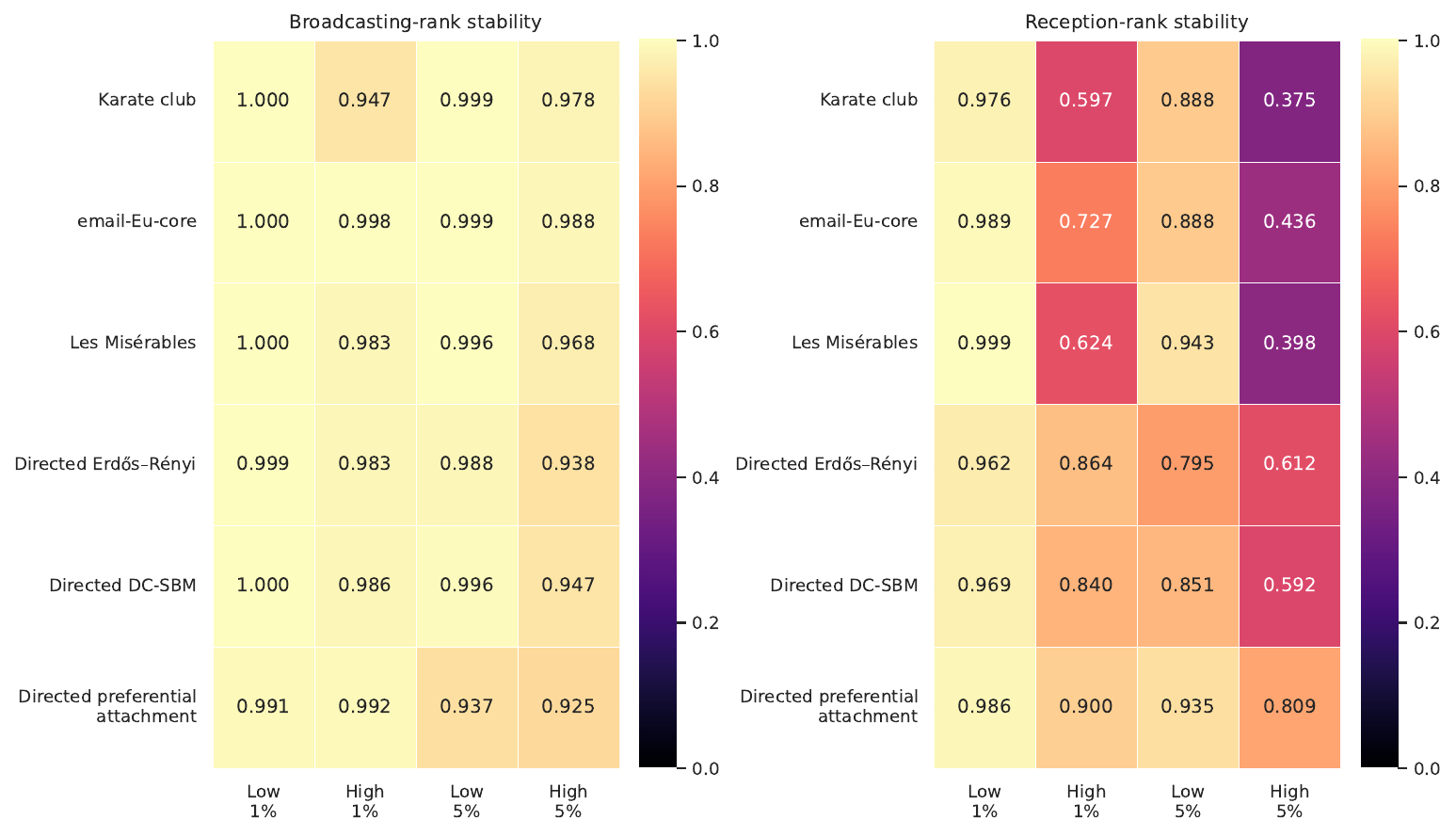}
\caption{Sensitivity of the mean broadcasting and reception rankings to neutral background boundary placement across the six original networks. ``Low'' and ``high'' denote the lowest and highest weighted out-strength vertices, respectively. Each annotation is the Spearman correlation with the source-only baseline after twelve baseline susceptibility profiles.}
\label{fig:reviewer2-boundary-stability}
\end{figure*}

Table~\ref{tab:reviewer2-performance} records the computational cost of the six-original-network alternate-distribution and boundary-placement robustness suite. Each solve computes one Green operator and therefore recovers all source responses for a single susceptibility and boundary condition through $R_{ab}=G_{ba}/G_{aa}$. Each of the six original networks required 84 recorded Green-operator solves, for 504 solves in total. On the largest network in that dense robustness suite, email-Eu-core ($n=986$), the 84 recorded solves required 8.653 seconds in total, with a mean solve time of 0.103009 seconds and maximum residual $2.22\times 10^{-15}$. These timings are hardware- and implementation-specific rather than asymptotic complexity claims. The expanded eight-network analysis uses the exact sparse method described in Section~\ref{subsec:evaluation-design}; its input hashes, solver residuals, and run metadata are recorded in the reproducibility materials.

\begin{table*}[p]
\centering
\caption{Observed computational cost of the six-original-network sensitivity analysis. A solve denotes one Green-operator computation for one susceptibility profile and boundary condition. Residuals are $\lVert (I-SW)G-I\rVert_{\max}$.}
\label{tab:reviewer2-performance}
\scriptsize
\begin{tabular}{lrrrrrr}
\hline
Network & $n$ & arcs & solves & mean s & total s & max residual  \\
\hline
Karate club & 34 & 156 & 84 & 0.000108 & 0.009 & 8.88e-16  \\
email-Eu-core & 986 & 24929 & 84 & 0.103009 & 8.653 & 2.22e-15  \\
Les Mis\'erables & 77 & 508 & 84 & 0.001209 & 0.102 & 6.66e-16  \\
Directed Erd\H{o}s--R\'enyi & 900 & 8917 & 84 & 0.098330 & 8.260 & 1.33e-15  \\
Directed DC-SBM & 900 & 9172 & 84 & 0.095602 & 8.031 & 1.67e-15  \\
Directed preferential attachment & 900 & 3172 & 84 & 0.094028 & 7.898 & 1.67e-15  \\
\hline
\end{tabular}
\end{table*}

\section*{Availability of data and materials}
The processed network inputs and numerical outputs supporting the conclusions of this article are archived at Zenodo: \url{https://doi.org/10.5281/zenodo.21716074}. Third-party source datasets are identified in the repository documentation and are obtained from their original providers.

\section*{Code availability}
The code used to generate the analyses, tables, and figures is archived at Zenodo (\url{https://doi.org/10.5281/zenodo.21716074} ) \citep{Boudourides2026FJCentralities}. The maintained source repository is available at \url{https://github.com/mboudour/BV_Friedkin_Johnsen_Infuence_Dynamics_and_Centralities}.

\section*{Declarations}

\subsection*{Ethics approval and consent to participate}

Not applicable.

\subsection*{Consent for publication}

Not applicable.

\subsection*{Competing interests}

The author declares that they have no competing interests.

\subsection*{Funding}

The author received no specific funding for this work.

\subsection*{Authors' contributions}

The author conceived the study, developed the theoretical results, designed and conducted the computational analyses, and wrote the manuscript.

\subsection*{Acknowledgements}

The author gratefully acknowledges Sergios Lenis for carrying out the first computations that initiated this project.

\bibliography{references}

\end{document}